\documentclass[english]{article}
\usepackage[T1]{fontenc}
\usepackage[latin9]{inputenc}
\usepackage{xcolor}
\usepackage{babel}
\usepackage{array}
\usepackage{float}
\usepackage{amsmath}
\usepackage{amssymb}
\usepackage{stmaryrd}
\usepackage{stackrel}
\usepackage{graphicx}
\usepackage[authoryear]{natbib}
\PassOptionsToPackage{normalem}{ulem}
\usepackage{ulem}
\usepackage[unicode=true,
 bookmarks=true,bookmarksnumbered=false,bookmarksopen=false,
 breaklinks=false,pdfborder={0 0 1},backref=false,colorlinks=false]
 {hyperref}
\hypersetup{pdftitle={\@shorttitle},
 pdfauthor={\@shortauthor},
 pdfsubject={\@Subject},
 pdfkeywords={Computer Graphics Forum, EUROGRAPHICS},
   colorlinks,linkcolor=blue,citecolor=blue,urlcolor=blue,    bookmarks=false,    pdfpagemode=UseNone}
\usepackage{breakurl}

\makeatletter

\providecommand{\tabularnewline}{\\}
\providecolor{lyxadded}{rgb}{0,0,1}
\providecolor{lyxdeleted}{rgb}{1,0,0}

\DeclareRobustCommand{\lyxsout}[1]{\ifx\\#1\else\sout{#1}\fi}

\usepackage{tikz,tkz-tab,tkz-graph}
\usepackage{amsthm}
\usepackage{tikz}
\usepackage{dsfont}

\newtheorem{prop}{Proposition}
\newtheorem{corol}{Corollary}
\newtheorem{defin}{Definition}
\newtheorem{theo}{Theorem}
\newtheorem{lemma}{Lemma}

\newtheorem{gage}{Guarantee}
\newtheorem{expectation}{Expectation}
\theoremstyle{plain}
\newtheorem{rmk}{Remark}

\usetikzlibrary{arrows,automata}
\usetikzlibrary{positioning, fadings}
\tikzset{
    state/.style={
           rectangle,
  fill=#1!5!white,
           rounded corners,
           draw=#1, very thick,
           minimum height=2em,
           inner sep=2pt,
           text centered,
           },
coeff/.style={
           circle,
           draw=black, very thick,
           minimum height=2em,
           inner sep=2pt,
           text centered,
           },	
}

\makeatother

\begin{document}

\title{Adjacency and Tensor Representation in General Hypergraphs{\normalsize{}}\\
{\normalsize{}Part 1: e-adjacency Tensor Uniformisation Using Homogeneous
Polynomials }}

\author{Xavier Ouvrard\textsuperscript{1,2}\quad{}Jean-Marie Le Goff\textsuperscript{1}\quad{}Stéphane
Marchand-Maillet\textsuperscript{2}\\
\\
1~:~CERN\qquad{}2~:~University of Geneva\\
{\small{}\{xavier.ouvrard\}@cern.ch}}
\maketitle
\begin{abstract}
Adjacency between two vertices in graphs or hypergraphs is a pairwise
relationship. It is redefined in this article as 2-adjacency. In general
hypergraphs, hyperedges hold for $n$-adic relationship. To keep the
$n$-adic relationship the concepts of k-adjacency and e-adjacency
are defined. In graphs 2-adjacency and e-adjacency concepts match,
just as $k$-adjacency and e-adjacency do for $k$-uniform hypergraphs.
For general hypergraphs these concepts are different. This paper also
contributes in a uniformization process of a general hypergraph to
allow the definition of an e-adjacency tensor, viewed as a hypermatrix,
reflecting the general hypergraph structure. This symmetric e-adjacency
hypermatrix allows to capture not only the degree of the vertices
and the cardinality of the hyperedges but also makes a full separation
of the different layers of a hypergraph.
\end{abstract}

\section{Introduction}

Hypergraphs were introduced in 1973 \citet{berge1973graphs}. Hypergraphs
have applications in many domains such as VLSI design, biology or
collaboration networks. Edges of a graph allow to connect only two
vertices where hyperedges of hypergraphs allow multiple vertices to
be connected. Recent improvements in tensor spectral theory have made
the research on the spectra of hypergraphs more relevant. For studying
such spectra a proper definition of general hypergraph Laplacian tensor
is needed and therefore the concept of adjacency has to be clearly
defined and consequently an (-as it will be defined later- e-)adjacency
tensor must be properly defined.

In \citet{pu2013relational} a clear distinction is made between the
pairwise relationship which is a binary relation and the co-occurrence
relationship which is presented as the extension of the pairwise relationship
to a $p$-adic relationship. The notion of co-occurrence is often
used in linguistic data as the simultaneous appearance of linguistic
units in a reference. The co-occurence concept can be easily extended
to vertices contained in a hyperedge: we designate it in hypergraphs
by the term e-adjacency.

Nonetheless it is more than an extension. Graph edges allow to connect
vertices by pair: graph adjacency concept is clearly a pairwise relationship.
At the same time in an edge only two vertices are linked. Also given
an edge only two vertices can be e-adjacent. Thus adjacency and e-adjacency
are equivalent in graphs.

Extending to hypergraphs the adjacency notion two vertices are said
adjacent if it exists a hyperedge that connect them. Hence the adjacency
notion still captures a binary relationship and can be modeled by
an adjacency matrix. But e-adjacency is no more a pairwise relationship
as a hyperedge being given more than two vertices can occur since
a hyperedge contains $p\geqslant1$ vertices. Therefore it is a $p$-adic
relationship that has to be captured and to be modeled by tensor.
Consequently adjacency matrix of a hypergraph and e-adjacency tensor
are two separated notions. Nonetheless the e-adjacency tensor if often
abusively named the adjacency tensor in the literature.

This article contributions are: 1. the definition of proper adjacency
concept in general hypergraphs; 2. a process to achieve the transformation
of a general hypergraph into a uniform hypergraph called uniformization
process; 3. the definition of a new (e-)adjacency tensor which not
solely preserves all the structural information of the hypergraph
but also captures separately the information on the hyperedges held
in the hypergraph.

After sketching the background and the related works on the adjacency
and e-adjacency concepts for hypergraphs in Section \ref{sec:Background-and-related},
one proposal is made to build a new e-adjacency tensor which is built
as unnormalized in Section \ref{sec:Towards-an-unnormalized}. Section
\ref{sec:Some-comments-on} tackles the particular case of graphs
seen as 2-uniform hypergraphs and the link with DNF. Future works
and Conclusion are addressed in Section \ref{sec:Future-work-and}.
A full example is given in Appendix A.

\subsubsection*{Notation}

Exponents are indicated into parenthesis - for instance $y^{(n)}$
- when they refer to the order of the corresponding tensor. Indices
are written into parenthesis when they refer of a sequence of objects
- for instance $a_{(k)ij}$ is the elements at row $i$ and column
$j$ of the matrix $A_{(k)}$ -. The context should made it clear.

For the convenience of readability, it is written $\boldsymbol{z}_{0}$
for $z^{1},...,z^{n}$. Hence given a polynomial $P$, $P\left(\boldsymbol{z}_{0}\right)$
has to be understood as $P\left(z^{1},...,z^{n}\right)$.

Given additional variables $y^{1},...,y^{k}$, it is written $\boldsymbol{z}_{k}$
for $z^{1},...,z^{n},y^{1},...,y^{k}$.

$\mathcal{S}_{k}$ is the set of permutations on the set $\left\{ i\,:\,i\in\mathbb{N}\land1\leqslant i\leqslant k\right\} $.

\section{Background and related works}

\label{sec:Background-and-related}

Several definitions of hypergraphs exist and are reminded in \citet{ouvrard2018hypergraph_survey}.
Hypergraphs allow the preservation of the $p$-adic relationship in
between vertices becoming the natural modeling of collaboration networks,
co-author networks, chemical reactions, genome and all situations
where the 2-adic relationship allowed by graphs is not sufficient
and where the keeping of the grouping information is important. Among
the existing definitions the one of \citet{bretto2013hypergraph}
is reminded:

\begin{defin}

An \textbf{(undirected)} \textbf{hypergraph} $\mathcal{H}=\left(V,E\right)$
on a finite set of $n$ vertices (or vertices) $V=\left\{ v_{1}\,,\,v_{2},\,...\,,\,v_{n}\right\} $
is defined as a family of $p$ \textbf{hyperedges} $E=\left\{ e_{1},e_{2},...,e_{p}\right\} $
where each hyperedge is a non-empty subset of $V$.

Let $\mathcal{H}=\left(V,E\right)$ be a hypergraph and $w$ a relation
such that each hyperedge $e\in E$ is mapped to a real number $w(e)$.
The hypergraph $\mathcal{H}_{w}=\left(V,E,w\right)$ is said to be
a \textbf{weighted hypergraph}.

The \textbf{$2$-section} of a hypergraph $\mathcal{H}=\left(V,E\right)$
is the graph $\left[\mathcal{H}\right]_{2}=\left(V,E'\right)$ such
that:
\[
\forall u\in V,\forall v\in V\,:\,(u,v)\in E'\Leftrightarrow\exists e\in E\,:\,u\in e\land v\in e
\]

Let $k\in\mathbb{N}^{*}$. a hypergraph is said to be \textbf{$k-$uniform}
if all its hyperedges have the same cardinality $k$. 

A \textbf{directed hypergraph} $\mathcal{H}=\left(V,E\right)$ on
a finite set of $n$ vertices (or vertices) $V=\left\{ v_{1}\,,\,v_{2},\,...\,,\,v_{n}\right\} $
is defined as a family of $p$ \textbf{hyperedges} $E=\left\{ e_{1},e_{2},...,e_{p}\right\} $
where each hyperedge contains exactly two non-empty subset of $V$,
one which is called the \textbf{source} - written $e_{s\,i}$ - and
the other one which is the \textbf{target} - written $e_{t\,i}$ -.

\end{defin}

In this article only undirected hypergraphs will be considered. In
a hypergraph a hyperedge links one or more vertices together. The
role of the hyperedges in hypergraphs is playing the role of edges
in graphs.

\begin{defin}

Let $\mathcal{H}=\left(V,E\right)$ be a hypergraph.

The \textbf{degree} of a vertex is the number of hyperedges it belongs
to. For a vertex $v_{i}$, it is written $d_{i}$ or $\deg\left(v_{i}\right)$.
It holds: $d_{i}=\left|\left\{ e\,:\,v_{i}\in e\right\} \right|$

\end{defin}

Given a hypergraph the incident matrix of an undirected hypergraph
is defined as follow:

\begin{defin}

The \textbf{incidence matrix} of a hypergraph is the rectangular matrix
$H=\left[h_{kl}\right]_{\substack{1\leqslant k\leqslant n\\
1\leqslant l\leqslant p
}
}$ of $M_{n\times p}\left(\left\{ 0\,;\,1\right\} \right)$, where $h_{kl}=\begin{cases}
1 & \text{if\,}\,\,v_{k}\in e_{l}\\
0 & \text{otherwise}
\end{cases}$.

\end{defin}

As seen in the introduction defining adjacency in a hypergraph has
to be distinguished from the e-adjacency in a hyperedge of a hypergraph.

\begin{defin}

Let $\mathcal{H}=\left(V,E\right)$ be a hypergraph. Let $u\in V$
and $v\in V$ be two vertices of this hypergraph.

$u$ and $v$ are said \textbf{adjacent} if it exists $e\in E$ such
that $u\in e$ and $v\in e$.

\end{defin}

\begin{defin}

Let $\mathcal{H}=\left(V,E\right)$ be a hypergraph. Let $k\geqslant1$
be an integer, $j\in\left\llbracket 1;k\right\rrbracket $, $i_{j}\in\left\llbracket 1;n\right\rrbracket $.
For $j\in\left\llbracket 1;k\right\rrbracket $, let $u_{i_{j}}\in V$
be $k$ vertices.

Then $u_{i_{1}}$,...,$u_{i_{k}}$ are said \textbf{$k$-adjacent}
if it exists $e\in E$ such that for all $j\in\left\llbracket 1;k\right\rrbracket $,
$u_{i_{j}}\in e$.

\end{defin}

With $k=2$, the usual notion of adjacency is retrieved.

If $k$ vertices are $k$-adjacent then each subset of this $k$ vertices
of size $l\leqslant k$ is $l$-adjacent.

\begin{defin}

Let $\mathcal{H}=\left(V,E\right)$ be a hypergraph. Let $e\in E$.

The vertices constituting $e$ are said \textbf{e-adjacent} vertices.

\end{defin}

If $\mathcal{H}$ is $k$-uniform then the $k$-adjacency is equivalent
to the e-adjacency of vertices in an edge.

For a general hypergraph, vertices that are $k$-adjacent with $k<\underset{e\in E}{\max}\left|e\right|$
have to co-occur - potentially with other vertices - in one edge.
In this case the notions of $k$-adjacency and of e-adjacency are
actually distinct.

\subsubsection*{Adjacency matrix}

The adjacency matrix of a hypergraph is related with the 2-adjacency.
Several approaches have been made to define an adjacency matrix for
hypergraphs.

In \citet{bretto2013hypergraph} the adjacency matrix is defined as:

\begin{defin}

The \textbf{adjacency matrix} is the square matrix which rows and
columns are indexed by the vertices of $\mathcal{H}$ and where for
all $u,v\in V$, $u\neq v$: $a_{uv}=\left|\left\{ e\in E\,:\,u,v\in e\right\} \right|$
and $a_{uu}=0$.

\end{defin}

The adjacency matrix is defined in \citet{zhou2007learning} as follow:

\begin{defin}

Let $\mathcal{H}_{w}=\left(V,E,w\right)$ be a weighted hypergraph.

The adjacency matrix of $\mathcal{H}_{w}$ is the matrix $A$ of size
$n\times n$ defined as 
\[
A=HWH^{T}-D_{v}
\]
 where $W$ is the diagonal matrix of size $p\times p$ containing
the weights of the hyperedges of $\mathcal{H}_{w}$ and $D_{v}$ is
the diagonal matrix of size $n\times n$ containing the degrees of
the vertices of $\mathcal{H}_{w}$, where $d(v)=\sum\limits _{\left\{ e\in E\,:v\in e\right\} }w(e)$
for all $v\in V$.

\end{defin}

This last definition is equivalent to the one of Bretto for unweighted
hypergraphs - ie weighted hypergraphs where the weight of all hyperedges
is 1.

The problem of the matrix approach is that the multi-adic relationship
is no longer kept as an adjacency matrix can link only pair of vertices.
Somehow it doesn't preserve the structure of the hypergraph: the hypergraph
is extended in the 2-section of the hypergraph and the information
is captured by this way. 

Following a lemma cited in \citet{dewar2016connectivity}, it can
be formulated:

\begin{lemma}

Let $\mathcal{H}=\left(V,E\right)$ be a hypergraph and let $u,v\in V$.
If two vertices $u$ and $v$ are adjacent in $\mathcal{H}$ then
they are adjacent in the 2-section $\left[\mathcal{H}\right]_{2}$.

\end{lemma}

The reciprocal doesn't hold as it would imply an isomorphism between
$\mathcal{H}$ and its 2-section $\left[\mathcal{H}\right]_{2}$.

Moving to the approach by e-adjacency will allow to keep the information
on the structure that is held in the hypergraph.

\subsubsection*{e-adjacency tensor}

In \citet{michoel2012alignment} an unnormalized version of the $k-$adjacency
tensor of a $k$-uniform hypergraph is given. This definition is also
adopted in \citet{ghoshdastidar2017uniform}.

\begin{defin}

The unnormalized ({[}Author's note{]}: $k$-)adjacency tensor of a
$k$-uniform hypergraph $\mathcal{H}=\left(V,E\right)$ on a finite
set of vertices $V=\left\{ v_{1}\,,\,v_{2},\,...\,,\,v_{n}\right\} $
and a family of hyperedges $E=\left\{ e_{1},e_{2},...,e_{p}\right\} $
of equal cardinality $k$ is the tensor $\mathcal{A}_{\text{raw}}=\left(a_{\text{raw}\,i_{1}...i_{k}}\right)_{1\leqslant i_{1},...,i_{k}\leqslant n}$
such that: 
\[
a_{\text{raw\,}i_{1}...i_{k}}=\begin{cases}
1 & \text{if }\left\{ v_{i_{1}},...,v_{i_{k}}\right\} \in E\\
0 & \text{otherwise.}
\end{cases}
\]

\end{defin}

In \citet{cooper2012spectra} a slightly different version exists
for the definition of the adjacency tensor, called the degree normalized
$k$-adjacency tensor

\begin{defin}

The ({[}Author's note{]}: \textbf{degree normalized $k$-})\textbf{adjacency
tensor} of a $k$-uniform hypergraph $\mathcal{H}=\left(V,E\right)$
on a finite set of vertices $V=\left\{ v_{1}\,,\,v_{2},\,...\,,\,v_{n}\right\} $
and a family of hyperedges $E=\left\{ e_{1},e_{2},...,e_{p}\right\} $
of equal cardinality $k$ is the tensor $\mathcal{A}=\left(a_{i_{1}...i_{k}}\right)_{1\leqslant i_{1},...,i_{k}\leqslant n}$
such that: 
\[
a_{i_{1}...i_{k}}=\dfrac{1}{(k-1)!}\begin{cases}
1 & \text{if }\left\{ v_{i_{1}},...,v_{i_{k}}\right\} \in E\\
0 & \text{otherwise.}
\end{cases}
\]

\end{defin}

This definition by introducing the coefficient $\dfrac{1}{\left(k-1\right)!}$
allows to retrieve the degree of a vertex $i$ summing the elements
of index $i$ on the first mode of the tensor. Also it will be called
the degree normalized adjacency tensor.

\begin{prop}

Let $\mathcal{H}=\left(V,E\right)$ be a $k$-uniform hypergraph.
Let $v_{i}\in V$ be a vertex. It holds by considering the degree
normalized $k-$adjacency tensor $\mathcal{A}=\left(a_{i_{1}...i_{k}}\right)_{1\leqslant i_{1},...,i_{k}\leqslant n}$:
\[
\deg\left(v_{i}\right)=\sum\limits _{i_{2},...,i_{k}=1}^{n}a_{ii_{2}...i_{k}}.
\]

\end{prop}

\begin{proof}

On the first mode of the degree normalized adjacency tensor, for a
given vertex $v_{i}$ that occurs in a hyperedge $e=\left\{ v_{i},v_{i_{2}},...,v_{i_{k}}\right\} $
the elements $a_{i\sigma\left(i_{2}\right)...\sigma\left(i_{k}\right)}=\dfrac{1}{(k-1)!}$
where $\sigma\in\mathcal{S}_{k-1}$ which exist in quantity $\left(k-1\right)!$
in the first mode. Hence $\sum\limits _{\sigma\in\mathcal{S}_{k-1}}a_{i\sigma\left(i_{2}\right)...\sigma\left(i_{k}\right)}=1$. 

Therefore doing it for all hyperedges where $v_{i}$ is an element
allows to retrieve the degree of $v_{i}$.

\end{proof}

This definition could be interpreted as the definition of the e-adjacency
tensor for a uniform hypergraph since the notion of $k$-adjacency
and e-adjacency are equivalent in a $k$-uniform hypergraph.

In \citet{hu2013spectral} a full study of the spectra of an uniform
hypergraph using the Laplacian tensor is given. The definition of
the Laplacian tensor is linked to the existence and definition of
the normalized ({[}Author's note{]}: $k$-)adjacency tensor.

\begin{defin}

The ({[}Author's note{]}: \textbf{eigenvalues}) \textbf{normalized}
({[}Author's note{]}: \textbf{$k$-})\textbf{adjacency tensor} of
a $k$-uniform hypergraph $\mathcal{H}=\left(V,E\right)$ on a finite
set of vertices $V=\left\{ v_{1}\,,\,v_{2},\,...\,,\,v_{n}\right\} $
and a family of hyperedges $E=\left\{ e_{1},e_{2},...,e_{p}\right\} $
of equal cardinality $k$ is the tensor $\overline{\mathcal{A}}=\left(a_{i_{1}...i_{k}}\right)_{1\leqslant i_{1},...,i_{k}\leqslant n}$
such that: 
\[
\overline{a_{i_{1}...i_{k}}}=\begin{cases}
\dfrac{1}{(k-1)!}\prod\limits _{1\leqslant j\leqslant k}\dfrac{1}{\sqrt[k]{d_{i_{j}}}} & \text{if }\left\{ v_{i_{1}},...,v_{i_{k}}\right\} \in E\\
0 & \text{otherwise.}
\end{cases}
\]

\end{defin}

The aim of the normalization is motivated by the bounding of the different
eigenvalues of the tensor.

The normalized Laplacian tensor $\mathcal{L}$ is given in the following
definition.

\begin{defin}

The normalized Laplacian tensor of a $k$-uniform hypergraph $\mathcal{H}=\left(V,E\right)$
on a finite set of vertices $V=\left\{ v_{1}\,,\,v_{2},\,...\,,\,v_{n}\right\} $
and a family of hyperedges $E=\left\{ e_{1},e_{2},...,e_{p}\right\} $
of equal cardinality $k$ is the tensor $\mathcal{L}=\mathcal{I}-\mathcal{A}$
where $\mathcal{I}$ is the $k$-th order $n$-dimensional diagonal
tensor with the $j$-th diagonal element $i_{j...j}=1$ if $d_{j}>0$
and 0 otherwise.

\end{defin}

In \citet{banerjee2017spectra} the definition is extended to general
hypergraph.

\begin{defin}Let $\mathcal{H}=\left(V,E\right)$ on a finite set
of vertices $V=\left\{ v_{1}\,,\,v_{2},\,...\,,\,v_{n}\right\} $
and a family of hyperedges $E=\left\{ e_{1},e_{2},...,e_{p}\right\} $.
Let $k_{\max}=\max\left\{ \left|e_{i}\right|:e_{i}\in E\right\} $
be the maximum cardinality of the family of hyperedges.

The adjacency hypermatrix of $\mathcal{H}$ written $\mathcal{A_{H}}=\left(a_{i_{1}...i_{k_{\max}}}\right)_{1\leqslant i_{1},...,i_{k_{\max}}\leqslant n}$
is such that for a hyperedge: $e=\left\{ v_{l_{1}},...,v_{l_{s}}\right\} $
of cardinality $s\leqslant k_{\max}$.

\[
a_{p_{1}...p_{k_{\max}}}=\dfrac{s}{\alpha}\text{, where }\alpha=\sum\limits _{\substack{k_{1},...,k_{s}\geqslant1\\
\sum k_{i}=k_{\text{max}}
}
}\dfrac{k_{\max}!}{k_{1}!...k_{s}!}
\]

with $p_{1}$, ..., $p_{k_{\max}}$ chosen in all possible way from
$\left\{ l_{1},...,l_{s}\right\} $ with at least once from each element
of $\left\{ l_{1},...,l_{s}\right\} $.

The other position of the hypermatrix are zero.

\end{defin}

The first problem in this case is that the notion of $k$-adjacency
as it has been mentioned earlier is not the most appropriated for
a general hypergraph where the notion of e-adjacency is much stronger.
The approach in \citet{shao2013general} and \citet{Pearson2014}
consists in the retrieval of the classical adjacency matrix for the
case where the hypergraph is 2-uniform - ie is a graph - by keeping
their degree invariant: therefore the degree of each vertex can be
retrieved on the first mode of the tensor by sum.

In \citet{hu2013spectral} the focus is made on the spectra of the
tensors obtained: the normalization is done to keep eigenvalues of
the tensor bounded. Extending this approach for general hypergraph,
\citet{banerjee2017spectra} spreads the information of lower cardinality
hyperedges inside the tensor. This approach focuses on the spectra
of the hypermatrix built. The e-adjacency cubical hypermatrix of order
$k_{\max}$ is kept at a dimension of the number of vertices $n$
at the price of splitting elements. Practically it could be hard to
use as the number of elements to be described for just one hyperedge
can explode. Indeed for each hyperedge the partition of $k_{\max}$
in $s$ parts has to be computed.

The number of partitions $p_{s}(m)$ of an integer $m$ in $s$ part
is given by the formula:

\[
p_{s}(m)=p_{s}(m-s)+p_{s-1}(m-1)
\]

This formula is obtained by considering the disjunctive case for splitting
$m$ in $s$ part:
\begin{itemize}
\item either the last part is equal to 1, and then $m-1$ has to be divided
in $s-1$;
\item or (exclusive) the $s$ parts are equals to at least 2, and then $m-s$
has to be divided in $s$.
\end{itemize}
First values of the number of partitions are given in Table \ref{Tab: Number of partitions of size s of an integer m}.

\begin{table}[H]
\begin{center}
\resizebox{0.9\textwidth}{!}{
\begin{tabular}{lrrrrrrrrrrrrrrrrrrrrrrrrr} \hline  m\textbackslash{}s & 1 &  2 &  3 &   4 &   5 &   6 &   7 &   8 &   9 &  10 &  11 &  12 & 13 & 14 & 15 & 16 & 17 & 18 & 19 & 20 & 21 & 22 & 23 & 24 & 25 \\  1   & 1 &  0 &  0 &   0 &   0 &   0 &   0 &   0 &   0 &   0 &   0 &   0 &  0 &  0 &  0 &  0 &  0 &  0 &  0 &  0 &  0 &  0 &  0 &  0 &  0 \\  2   & 1 &  1 &  0 &   0 &   0 &   0 &   0 &   0 &   0 &   0 &   0 &   0 &  0 &  0 &  0 &  0 &  0 &  0 &  0 &  0 &  0 &  0 &  0 &  0 &  0 \\  3   & 1 &  1 &  1 &   0 &   0 &   0 &   0 &   0 &   0 &   0 &   0 &   0 &  0 &  0 &  0 &  0 &  0 &  0 &  0 &  0 &  0 &  0 &  0 &  0 &  0 \\  4   & 1 &  2 &  1 &   1 &   0 &   0 &   0 &   0 &   0 &   0 &   0 &   0 &  0 &  0 &  0 &  0 &  0 &  0 &  0 &  0 &  0 &  0 &  0 &  0 &  0 \\  5   & 1 &  2 &  2 &   1 &   1 &   0 &   0 &   0 &   0 &   0 &   0 &   0 &  0 &  0 &  0 &  0 &  0 &  0 &  0 &  0 &  0 &  0 &  0 &  0 &  0 \\  6   & 1 &  3 &  3 &   2 &   1 &   1 &   0 &   0 &   0 &   0 &   0 &   0 &  0 &  0 &  0 &  0 &  0 &  0 &  0 &  0 &  0 &  0 &  0 &  0 &  0 \\  7   & 1 &  3 &  4 &   3 &   2 &   1 &   1 &   0 &   0 &   0 &   0 &   0 &  0 &  0 &  0 &  0 &  0 &  0 &  0 &  0 &  0 &  0 &  0 &  0 &  0 \\  8   & 1 &  4 &  5 &   5 &   3 &   2 &   1 &   1 &   0 &   0 &   0 &   0 &  0 &  0 &  0 &  0 &  0 &  0 &  0 &  0 &  0 &  0 &  0 &  0 &  0 \\  9   & 1 &  4 &  7 &   6 &   5 &   3 &   2 &   1 &   1 &   0 &   0 &   0 &  0 &  0 &  0 &  0 &  0 &  0 &  0 &  0 &  0 &  0 &  0 &  0 &  0 \\  10  & 1 &  5 &  8 &   9 &   7 &   5 &   3 &   2 &   1 &   1 &   0 &   0 &  0 &  0 &  0 &  0 &  0 &  0 &  0 &  0 &  0 &  0 &  0 &  0 &  0 \\  11  & 1 &  5 & 10 &  11 &  10 &   7 &   5 &   3 &   2 &   1 &   1 &   0 &  0 &  0 &  0 &  0 &  0 &  0 &  0 &  0 &  0 &  0 &  0 &  0 &  0 \\  12  & 1 &  6 & 12 &  15 &  13 &  11 &   7 &   5 &   3 &   2 &   1 &   1 &  0 &  0 &  0 &  0 &  0 &  0 &  0 &  0 &  0 &  0 &  0 &  0 &  0 \\  13  & 1 &  6 & 14 &  18 &  18 &  14 &  11 &   7 &   5 &   3 &   2 &   1 &  1 &  0 &  0 &  0 &  0 &  0 &  0 &  0 &  0 &  0 &  0 &  0 &  0 \\  14  & 1 &  7 & 16 &  23 &  23 &  20 &  15 &  11 &   7 &   5 &   3 &   2 &  1 &  1 &  0 &  0 &  0 &  0 &  0 &  0 &  0 &  0 &  0 &  0 &  0 \\  15  & 1 &  7 & 19 &  27 &  30 &  26 &  21 &  15 &  11 &   7 &   5 &   3 &  2 &  1 &  1 &  0 &  0 &  0 &  0 &  0 &  0 &  0 &  0 &  0 &  0 \\  16  & 1 &  8 & 21 &  34 &  37 &  35 &  28 &  22 &  15 &  11 &   7 &   5 &  3 &  2 &  1 &  1 &  0 &  0 &  0 &  0 &  0 &  0 &  0 &  0 &  0 \\  17  & 1 &  8 & 24 &  39 &  47 &  44 &  38 &  29 &  22 &  15 &  11 &   7 &  5 &  3 &  2 &  1 &  1 &  0 &  0 &  0 &  0 &  0 &  0 &  0 &  0 \\  18  & 1 &  9 & 27 &  47 &  57 &  58 &  49 &  40 &  30 &  22 &  15 &  11 &  7 &  5 &  3 &  2 &  1 &  1 &  0 &  0 &  0 &  0 &  0 &  0 &  0 \\  19  & 1 &  9 & 30 &  54 &  70 &  71 &  65 &  52 &  41 &  30 &  22 &  15 & 11 &  7 &  5 &  3 &  2 &  1 &  1 &  0 &  0 &  0 &  0 &  0 &  0 \\  20  & 1 & 10 & 33 &  64 &  84 &  90 &  82 &  70 &  54 &  42 &  30 &  22 & 15 & 11 &  7 &  5 &  3 &  2 &  1 &  1 &  0 &  0 &  0 &  0 &  0 \\  21  & 1 & 10 & 37 &  72 & 101 & 110 & 105 &  89 &  73 &  55 &  42 &  30 & 22 & 15 & 11 &  7 &  5 &  3 &  2 &  1 &  1 &  0 &  0 &  0 &  0 \\  22  & 1 & 11 & 40 &  84 & 119 & 136 & 131 & 116 &  94 &  75 &  56 &  42 & 30 & 22 & 15 & 11 &  7 &  5 &  3 &  2 &  1 &  1 &  0 &  0 &  0 \\  23  & 1 & 11 & 44 &  94 & 141 & 163 & 164 & 146 & 123 &  97 &  76 &  56 & 42 & 30 & 22 & 15 & 11 &  7 &  5 &  3 &  2 &  1 &  1 &  0 &  0 \\  24  & 1 & 12 & 48 & 108 & 164 & 199 & 201 & 186 & 157 & 128 &  99 &  77 & 56 & 42 & 30 & 22 & 15 & 11 &  7 &  5 &  3 &  2 &  1 &  1 &  0 \\  25  & 1 & 12 & 52 & 120 & 192 & 235 & 248 & 230 & 201 & 164 & 131 & 100 & 77 & 56 & 42 & 30 & 22 & 15 & 11 &  7 &  5 &  3 &  2 &  1 &  1 \\ \hline \end{tabular}
}
\end{center}

\caption{Number of partitions of size $s$ of an integer $m$ }
\label{Tab: Number of partitions of size s of an integer m}
\end{table}

This number of partition gives the number of elements to be specified
for a single hyperedge in the Banerjee's hypermatrix, as they can't
be obtained directly by permutation. This number varies depending
on the cardinality of the hyperedge to be represented. This variation
is not a monotonic function of the size $s$.

The value of $\alpha$ to be used for a given hyperedge of size $s$
for a maximal cardinality $k_{\max}$ of the Banerjee's adjacency
tensor is given in Table \ref{Tab: Nb of elements to be filled for an hyperedge of size s}.
This value also reflects the number of elements to be filled in the
hypermatrix for a single hyperedge.

\begin{table}[H]
\begin{center}
\resizebox{0.9\textwidth}{!}{
\begin{tabular}{llllllllr} \hline  $s\textbackslash{}k_{\max}$ & 5  & 10     & 15         & 20             & 25                  & 30                        & 35                             &                                   40 \\  1     & 1  & 1      & 1          & 1              & 1                   & 1                         & 1                              &                                    1 \\  2     & 30 & 896    & 32766      & 956196         & 33554430            & 996183062                 & 34359738366                    &                        1030588363364 \\  3     & 50 & 23640  & 6357626    & 1553222032     & 382505554925        & 94743186241770            & 22960759799383757              &                  5611412540548420920 \\  4     & 20 & 100970 & 135650970  & 149440600896   & 158221556736195     & 164769169326140215        & 170721045139376180665          &             176232934305968169141592 \\  5     & 1  & 125475 & 745907890  & 2826175201275  & 9506452442642751    & 30773997163632534765      & 98200286674992772689630        &          311409618017926342757598795 \\  6     & x  & 61404  & 1522977456 & 17420742448158 & 158199194804672560  & 1322183126915502403463    & 10690725777258446036242741     &        85180421514142371562050204468 \\  7     & x  & 14280  & 1425364941 & 46096037018576 & 1024206402004025515 & 19673349126500416962615   & 354878263731993584768297882    &      6217590037131694711658104802268 \\  8     & x  & 1500   & 702714870  & 61505129881418 & 3154352367940801390 & 129129229794015955874175  & 4769303064589903155918576810   &    167503457011878955780131372020240 \\  9     & x  & 90     & 201328985  & 46422598935960 & 5267776889834101885 & 437004824231068745652585  & 31134364616525428333788664160  &   2051990575671846572076732402739560 \\  10    & x  & 1      & 35145110   & 21559064120035 & 5237969253953146975 & 848748719343315752120887  & 111787775515270562752918708505 &  13174986533143342163734795019830855 \\  11    & x  & x      & 3709706    & 6508114071602  & 3332426908789146245 & 1023444669605845490919630 & 241305539520076885874877723856 &  49059583248616094623568196287767720 \\  12    & x  & x      & 242970     & 1320978392032  & 1430090837664465640 & 814611609439944701336120  & 334883841129942857103836783480 & 114204835945488341535343378586826510 \\  13    & x  & x      & 9100       & 184253421690   & 429168957710189920  & 448888886709990497395170  & 315061943784480485752922317100 & 176097407919167018972821102617824800 \\  14    & x  & x      & 210        & 17758229920    & 92361393090110900   & 177434686702809581360280  & 209636307340035341769456805590 & 188390878586504393731248560781565540 \\  15    & x  & x      & 1          & 1182354420     & 14515221518630650   & 51629112999502425355050   & 101972261667580282621340734042 & 145207225656117240323230829098848300 \\  16    & x  & x      & x          & 53422908       & 1686842411440120    & 11274940758810423952590   & 37193647457294620660325206920  &  83124043946911069759380261652009018 \\  17    & x  & x      & x          & 1637610        & 145857986021220     & 1875745279587180337830    & 10373941738039097562798529130  &  36202281770971401316508548887148260 \\  18    & x  & x      & x          & 31350          & 9387370139400       & 240458041631247630090     & 2247098355408068243367808830   &  12227164493902961371079076114591450 \\  19    & x  & x      & x          & 380            & 446563570200        & 23950282001673975675      & 382710033315178514982029070    &   3252386812566620163782349432515670 \\  20    & x  & x      & x          & 1              & 15571428950         & 1862767268307916425       & 51758773473472067323039950     &    690009783002559481444810135863737 \\  21    & x  & x      & x          & x              & 390169010           & 113301447816411855        & 5602215923984438576703270      &    117978632939681392614390018854490 \\  22    & x  & x      & x          & x              & 6932200             & 5375646410875455          & 488160287033902614520290       &     16396955289494938961248184877710 \\  23    & x  & x      & x          & x              & 80500               & 197788491523350           & 34380160285907377001220        &      1865425003253790074111730106860 \\  24    & x  & x      & x          & x              & 600                 & 5587457302050             & 1960619958296697461400         &       174704650201012418163972506640 \\  25    & x  & x      & x          & x              & 1                   & 119813107050              & 90483896754284001150           &        13528775872638975527061789150 \\  26    & x  & x      & x          & x              & x                   & 1909271637                & 3368998127887283892            &          868981935345151947003947262 \\  27    & x  & x      & x          & x              & x                   & 22143240                  & 100617182607307212             &           46381804383191991754075704 \\  28    & x  & x      & x          & x              & x                   & 172550                    & 2391172870380140               &            2057782621039570457724152 \\  29    & x  & x      & x          & x              & x                   & 870                       & 44721107569820                 &              75781801182259804328840 \\  30    & x  & x      & x          & x              & x                   & 1                         & 649591878320                   &               2309066362145733662940 \\  31    & x  & x      & x          & x              & x                   & x                         & 7166900664                     &                 57915248685968404016 \\  32    & x  & x      & x          & x              & x                   & x                         & 58538480                       &                  1187293166698640716 \\  33    & x  & x      & x          & x              & x                   & x                         & 327250                         &                    19717915340636370 \\  34    & x  & x      & x          & x              & x                   & x                         & 1190                           &                      262203877675610 \\  35    & x  & x      & x          & x              & x                   & x                         & 1                              &                        2751867046110 \\  36    & x  & x      & x          & x              & x                   & x                         & x                              &                          22273515966 \\  37    & x  & x      & x          & x              & x                   & x                         & x                              &                            135074420 \\  38    & x  & x      & x          & x              & x                   & x                         & x                              &                               568100 \\  39    & x  & x      & x          & x              & x                   & x                         & x                              &                                 1560 \\  40    & x  & x      & x          & x              & x                   & x                         & x                              &                                    1 \\ \hline 
\end{tabular}
}
\end{center}

\caption{Number of elements to be filled for a hyperedge of size $s$ given
maximal cardinality $m$}
\label{Tab: Nb of elements to be filled for an hyperedge of size s}
\end{table}

In this article, the proposed method to elaborate an e-adjacency tensor
focuses on the interpretability of the construction: a uniformization
process is proposed in which a general hypergraph is transformed in
a uniform hypergraph by adding to it elements. The strong link made
with homogeneous polynomials reinforce the choice made and allow to
retrieve proper matrix of a uniform hypergraph at the end of the process.
The additional vertices help to capture not solely the e-adjacency
but also give the ability to hold the $k$-adjacency whatever the
level it occurs.

The approach is based on the homogeneisation of sums of polynomials
of different degrees and by considering a family of uniform hypergraphs.
It is also motivated by the fact that the information on the cardinality
of the hyperedges has to be kept in some ways and that the elements
should not be mixed with the different layers of the hypergraph.

\section{Towards an e-adjacency tensor of a general hypergraph}

\label{sec:Towards-an-unnormalized}

To build an e-adjacency tensor for a general hypergraph we need a
way to store elements which represent the hyperedges. As these hyperedges
have different cardinalities, the representation of the e-adjacency
of vertices in a unique tensor can be achieved only by filling the
hyperedges with additional elements. The problem of finding an e-adjacency
tensor of a general hypergraph is then transformed in a uniformization
problem.

This uniformisation process should be at least interpretable in term
of uniform hypergraphs. It should capture the structural information
of the hypergraph, which includes information on number of hyperedges,
degrees of vertices and additional information on the profile of the
hypergraph.

We propose a framework based on homogeneous polynomials that are iteratively
summed by weighting with technical coefficients and homogeneized.
This uniformisation process allows the construction of a weighted
uniform hypergraph. The technical coefficients are adjusted to allow
the handshake lemma to hold in the built uniform hypergraph.

\subsection{Family of tensors attached to a hypergraph}

Let $\mathcal{H}=\left(V,E\right)$ be a hypergraph. A hypergraph
can be decomposed in a family of uniform hypergraphs. To achieve it,
let $\mathcal{R}$ be the equivalency relation: $e\mathcal{R}e'\Leftrightarrow\left|e\right|=\left|e'\right|$.

$E/\mathcal{R}$ is the set of classes of hyperedges of same cardinality.
The elements of $E/\mathcal{R}$ are the sets: $E_{k}=\left\{ e\in E:\,\left|e\right|=k\right\} $.

Let $k_{\max}=\max\limits _{e\in E}\left|e\right|$, called the \textbf{range}
of the hypergraph $\mathcal{H}$

Considering $K=\left\{ k:E_{k}\in E/\mathcal{R}\right\} $, it is
set $E_{k}=\emptyset$ for all $k\in\left\llbracket 1\,;\,k_{\max}\right\rrbracket \backslash K$.

Let consider the hypergraphs: $\mathcal{H}_{k}=\left(V,E_{k}\right)$
for all $k\in\left\llbracket 1\,;\,k_{\max}\right\rrbracket $ which
are all $k$-uniform.

It holds: $E=\bigcup\limits _{k=1}^{k_{\max}}E_{k}$ and $E_{j}\cap E_{k}=\emptyset$
for all $j\neq k$, hence $\left(E_{k}\right)_{1\leqslant k\leqslant k_{\text{\text{max}}}}$
formed a partition of $E$ which is unique by the way it has been
defined.

Before going forward the sum of two hypergraphs has to be defined:

\begin{defin}

Let $\mathcal{H}_{1}=\left(V_{1},E_{1}\right)$ and $\mathcal{H}_{2}=\left(V_{2},E_{2}\right)$
be two hypergraphs. The sum of these two hypergraphs is the hypergraph
written $\mathcal{H}_{1}+\mathcal{H}_{2}$ defined as: 
\[
\mathcal{H}_{1}+\mathcal{H}_{2}=\left(V_{1}\cup V_{2},E_{1}\cup E_{2}\right).
\]

This sum is said direct if $E_{1}\cap E_{2}=\emptyset$. In this case
the sum is written $\mathcal{H}_{1}\oplus\mathcal{H}_{2}$.

\end{defin}

Hence: 
\[
\mathcal{H}=\bigoplus\limits _{k=1}^{k_{\max}}\mathcal{H}_{k}.
\]

The hypergraph $\mathcal{H}$ is said to be decomposed in a family
of hypergraphs $\left(\mathcal{H}_{k}\right)_{1\leqslant k\leqslant k_{\text{\text{max}}}}$
where $\mathcal{H}_{k}$ is $k$-uniform. 

An illustration of the decomposition of a hypergraph in a family of
uniform hypergraphs is shown in Figure \ref{Fig: Layers of an hypergraph}.
This family of uniform hypergraphs decomposes the original hypergraph
in layers. A layer holds a $k$-uniform hypergraph ($1\leqslant k\leqslant k_{\max}$):
therefore the layer is said to be of level $k$.

\begin{figure}
\begin{center}\includegraphics[scale=0.3]{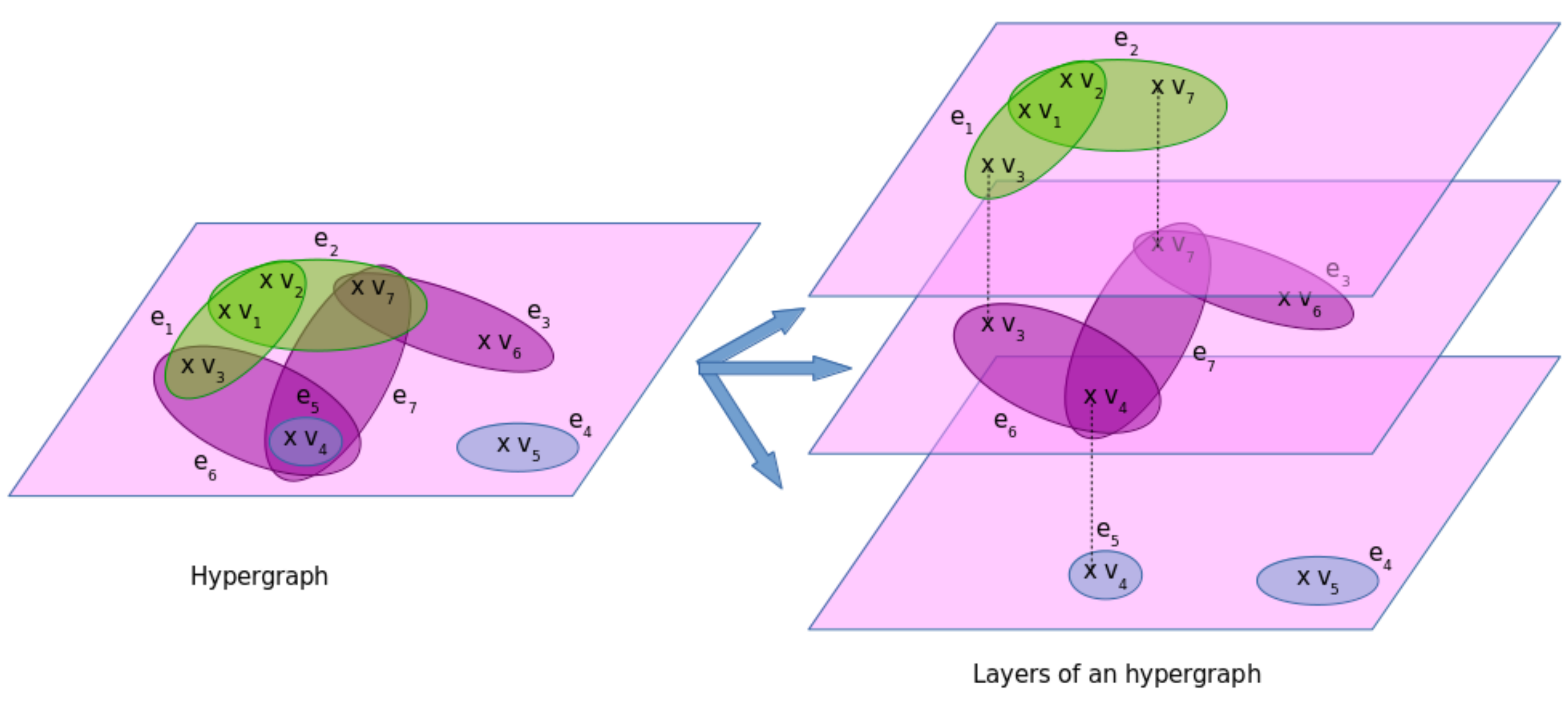}\end{center}

\caption{Illustration of a hypergraph decomposed in three layers of uniform
hypergraphs}
\label{Fig: Layers of an hypergraph}
\end{figure}

Therefore, at each $k$-uniform hypergraph $\mathcal{H}_{k}$ can
be mapped a ($k$-adjacency) e-adjacency tensor $\mathcal{A}_{k}$
of order $k$ which is hypercubic and symmetric of dimension $\left|V\right|=n$.
This tensor can be unnormalized or normalized.

Choosing one type of tensor - normalized or unnormalized for the whole
family of $\mathcal{H}_{k}$ - the hypergraph $\mathcal{H}$ is then
fully described by the family of e-adjacency tensors $\mathcal{A}_{\mathcal{H}}=\left(\mathcal{A}_{k}\right)$.
In the case where all the $\mathcal{A}_{k}$ are chosen normalized
this family is said pre-normalized. The final choice will be made
further in Sub-Section \ref{subsec:Gages-of-the} and explained to
fullfill the expectations listed in the next Sub-Section.

\subsection{Expectations for an e-adjacency tensor for a general hypergraph}

\label{subsec:Expectations-for-a}

The definition of the family of tensors attached to a general hypergraph
has the advantage to open the way to the spectral theory for uniform
hypergraphs which is quite advanced.

Nonetheless many problems remain in keeping a family of tensors of
different orders: studying the spectra of the whole hypergraph could
be hard to achieve by this means. Also it is necessary to get an
e-adjacency tensor which covers the whole hypergraph and which retains
the information on the whole structure.

The idea behind is to ``fill'' the hyperedges with sufficient elements
such that the general hypergraph is transformed in an uniform hypergraph
through a uniformisation process. A similar approach has been taken
in \citet{banerjee2017spectra} where the filling elements are the
vertices belonging to the hyperedge itself. In the next subsections
the justification of the approach taken will be made via homogeneous
polynomials. Before getting to the construction, expected properties
of such a tensor have to be listed.

\begin{expectation}

The tensor should be symmetric and its generation should be simple.

\end{expectation}

This expectation emphasizes the fact that in between two built e-adjacency
tensor, the one that can be easily generated has to be chosen: it
includes the fact that the tensor has to be described in a simple
way.

\begin{expectation}

The tensor should be invariant to vertices permutation either globally
or at least locally.

\end{expectation}

This expectation is motivated by the fact that in a hyperedge the
vertices have no order. The fact that this expectation can be local
remains in the fact that added special vertices will not have the
same status that the one from the original hypergraph. Also the invariance
by permutation is expected on the vertices of the original hypergraph.

\begin{expectation}

The e-adjacency tensor should allow the retrieval of the hypergraph
it is originated from.

\end{expectation}

This expectation seems important to rebuild properly the original
hypergraph from the e-adjacency tensor: all the necessary information
to retrieve the original hyperedges has to be encoded in the tensor.

\begin{expectation}

Giving the choice of two representations the sparsest e-adjacency
tensor should be chosen.

\end{expectation}

Sparsity allows to compress the information and so to gain in place
and complexity in calculus. Also sparsity is a desirable property
for some statistical reasons as shown in \citet{nikolova2000} or
expected in \citet{bruckstein2009sparse} for signal processing and
image encoding.

\begin{expectation}

The e-adjacency tensor should allow the retrieval of the vertex degrees.

\end{expectation}

In the adjacency matrix of a graph the information on the degrees
of the vertices is encoded directly. It is still the case, as it has
been seen, with the $k$-adjacency degree normalised tensor that has
been defined by \citet{shao2013general} and \citet{Pearson2014}.

\subsection{Tensors family and homogeneous polynomials family}

To construct an homogeneous polynomial representing a general hypergraph,
the family of e-adjacency tensors obtained in the previous Subsection
is mapped to a family of homogeneous polynomials. This mapping is
used in \citet{comon2015polynomial} where the author links symmetric
tensors and homogeneous polynomials of degree $s$ to show that the
problem of the CP decomposition of different symmetric tensors of
different orders and the decoupled representation of multivariate
polynomial maps are related.

\subsubsection{Homogeneous polynomials family of a hypergraph}

Let $\mathbb{K}$ be a field. Here $\mathbb{K}=\mathbb{R}$.

Let $\mathcal{A}_{k}\in\mathcal{L}_{k}^{0}\left(\mathbb{K}^{n}\right)$
be a cubical tensor of order $k$ and dimension $n$ with values in
$\mathbb{K}$.

\begin{defin}

Let define the \textbf{Segre outerproduct} $\otimes$ of $\boldsymbol{a}=\left[a_{i}\right]\in\mathbb{K}^{l}$
and $\boldsymbol{b}=\left[b_{j}\right]\in\mathbb{K}^{m}$ as: 
\[
\boldsymbol{a}\otimes\boldsymbol{b}=\left[a_{i}b_{j}\right]_{\substack{1\leqslant i\leqslant l\\
1\leqslant j\leqslant m
}
}\in\mathbb{K}^{l\times m}.
\]

More generaly as given in \citet{comon2008symmetric} the outerproduct
of $k$ vectors $u_{(1)}\in\mathbb{K}^{n_{1}}$, ..., $u_{(k)}\in\mathbb{K}^{n_{k}}$
is defined as:

\[
\stackrel[i=1]{k}{\otimes}\boldsymbol{u_{(i)}}=\left[\prod\limits _{i=1}^{k}u_{(i)j_{i}}\right]_{j_{1},..,j_{k}=1}^{n_{1},...,n_{k}}\in\mathbb{K}^{n_{1}\times...\times n_{k}}.
\]

\end{defin}

Let $\boldsymbol{e_{1}}$, $...$, $\boldsymbol{e_{n}}$ be the canonical
basis of $\mathbb{K}^{n}$.

$\left(\boldsymbol{e_{i_{1}}}\otimes...\otimes\boldsymbol{e_{i_{k}}}\right)_{1\leqslant i_{1},...,i_{k}\leqslant n}$
is a basis of $\mathcal{L}_{k}^{0}\left(\mathbb{K}^{n}\right)$.

Then $\mathcal{A}_{k}$ can be written as:

\[
\mathcal{A}_{k}=\sum\limits _{1\leqslant i_{1},...,i_{k}\leqslant n}a_{(k)\,i_{1}...i_{k}}\boldsymbol{e_{i_{1}}}\otimes...\otimes\boldsymbol{e_{i_{k}}}
\]

The notation $\boldsymbol{A_{k}}$ will be used for the corresponding
hypermatrix of coefficients $a_{(k)\,i_{1}...i_{k}}$ where $1\leqslant i_{1},...,i_{k}\leqslant n$.

Let $z\in\mathbb{K}^{n}$, with $z=z^{i}\boldsymbol{e_{i}}$ using
the Einstein convention.

In \citet{lim2013tensors} a multilinear matrix multiplication is
defined as follow:

\begin{defin}

Let $\boldsymbol{A}\in\mathbb{K}^{n_{1}\times...\times n_{d}}$ and
$X_{j}=\left[x_{(j)kl}\right]\in\mathbb{K}^{m_{j}\times n_{j}}$ for
$1\leqslant j\leqslant d$.

$\boldsymbol{A^{\prime}}=\left(X_{1},...,X_{d}\right).\boldsymbol{A}$
is the multilinear matrix multiplication and defined as the matrix
of $\mathbb{K}^{m_{1}\times...\times m_{d}}$ of coefficients: 
\[
a_{j_{1}...j_{d}}^{\prime}=\sum\limits _{k_{1},...,k_{d}=1}^{n_{1},...,n_{d}}x_{(1)j_{1}k_{1}}...x_{(d)j_{d}k_{d}}a_{k_{1}...k_{d}}
\]
 for $1\leqslant j_{i}\leqslant m_{i}$ with $1\leqslant i\leqslant d$.

\end{defin}

Afterwards only vectors $z\in\mathbb{K}^{n}$ are needed and $\mathcal{A}_{k}$
is cubical of order $k$ and dimension $n$. Writing $\left(z,...,z\right)\in\left(\mathbb{K}^{n}\right)^{k}$,
$\left(z\right)_{[k]}$

Therefore $(z)_{[k]}.\boldsymbol{A_{k}}$ contains only one element
written $P_{k}\left(z^{1},...,z^{n}\right)=P_{k}\left(\boldsymbol{z}_{0}\right)$:
\begin{equation}
P_{k}\left(\boldsymbol{z}_{0}\right)=\sum\limits _{1\leqslant i_{1},...,i_{k}\leqslant n}a_{(k)\,i_{1}...i_{k}}z^{i_{1}}...z^{i_{k}}.\label{eq:P_k unreduced}
\end{equation}
\label{Eq:1}

Therefore considering a hypergraph $\mathcal{H}$ with its family
of unnormalized tensor $\mathcal{A}_{\mathcal{H}}=\left(\mathcal{A}_{k}\right)$,
it can be also attached a family $P_{\mathcal{H}}=\left(P_{k}\right)$
of homogenous polynomials with $\deg\left(P_{k}\right)=k$.

The formulation of $P_{k}$ can be reduced taking into account that
$\mathcal{A}_{k}$ is symmetric for a uniform hypergraph:

\begin{equation}
P_{k}\left(\boldsymbol{z}_{0}\right)=\sum\limits _{1\leqslant i_{1}\leqslant...\leqslant i_{k}\leqslant n}k!a_{(k)\,i_{1}...i_{k}}z^{i_{1}}...z^{i_{k}}.\label{eq:P_k reduced}
\end{equation}

Writing:

\begin{equation}
\widetilde{P_{k}}\left(\boldsymbol{z}_{0}\right)=\sum\limits _{1\leqslant i_{1}\leqslant...\leqslant i_{k}\leqslant n}\alpha_{(k)\,i_{1}...i_{k}}z^{i_{1}}...z^{i_{k}}.\label{eq:P_k reduced_v2}
\end{equation}

the reduced form of $P_{k}$, it holds:

\[
P_{k}\left(\boldsymbol{z}_{0}\right)=k!\widetilde{P_{k}}\left(\boldsymbol{z}_{0}\right).
\]

Writing for $1\leqslant i_{1}\leqslant...\leqslant i_{k}\leqslant n$:
\[
\alpha_{(k)\,i_{1}...i_{k}}=k!a_{(k)\,i_{1}...i_{k}}
\]

and $\alpha_{(k)\,\sigma\left(i_{1}\right)...\sigma\left(i_{k}\right)}=0$
for $\sigma\in\mathcal{S}_{k}$, $\sigma\neq\text{Id}$

It holds:

\begin{eqnarray}
P_{k}\left(\boldsymbol{z}_{0}\right) & = & \sum\limits _{1\leqslant i_{1}\leqslant...\leqslant i_{k}\leqslant n}\alpha_{(k)\,i_{1}...i_{k}}z^{i_{1}}...z^{i_{k}}\nonumber \\
 & = & \sum\limits _{1\leqslant i_{1},...,i_{k}\leqslant n}\alpha_{(k)\,i_{1}...i_{k}}z^{i_{1}}...z^{i_{k}}.\label{eq:P_k reduced_v3}
\end{eqnarray}

and:

\begin{equation}
\widetilde{P_{k}}\left(\boldsymbol{z}_{0}\right)=\sum\limits _{1\leqslant i_{1}\leqslant...\leqslant i_{k}\leqslant n}\dfrac{\alpha_{(k)\,i_{1}...i_{k}}}{k!}z^{i_{1}}...z^{i_{k}}.\label{eq:P_k reduced_v2-1}
\end{equation}

\subsubsection{Reversibility of the process}

Reciprocally, given a homogeneous polynomial of degree $k$ a unique
hypercubic tensor of order $k$ can be built: its dimension is the
number of different variables in the homogeneous polynomial. If the
homogeneous polynomial of degree $k$ is supposed reduced and ordered
then only one hypercubic and symmetric hypermatrix can be built. It
reflects uniquely a $k$-uniform hypergraph adding the constraint
that each monomial is composed of the product of $k$ different variables.

\begin{prop}

Let $P\left(z_{0}\right)=\sum\limits _{1\leqslant i_{1},...,i_{k}\leqslant n}a_{i_{1}...i_{k}}z^{i_{1}}...z^{i_{k}}$
be a homogeneous polynomial of degree $k$ where:
\begin{itemize}
\item for $j\neq k$ : $z^{j}\neq z^{k}$ 
\item for all $1\leqslant j\leqslant n$: $\deg\left(z^{j}\right)=1$
\item and such that for all $\sigma\in\mathcal{S}_{k}$: $a_{\sigma\left(i_{1}\right)...\sigma\left(i_{k}\right)}=a_{i_{1}...i_{k}}$.
\end{itemize}
Then $P$ is the homogeneous polynomial attached to a unique $k$-uniform
hypergraph $\mathcal{H}=\left(V,E,w\right)$ - up to the indexing
of vertices.

\end{prop}

\begin{proof}

Considering the vertices $\left(v_{i}\right)_{1\leqslant i\leqslant n}$
labellized by the elements of $\left\llbracket 1,n\right\rrbracket $.

If $a_{i_{1}...i_{k}}\neq0$ then for all $\sigma\in\mathcal{S}_{k}$:
$a_{\sigma\left(i_{1}\right)...\sigma\left(i_{k}\right)}$ a unique
hyperedge $e_{j}$ is attached corresponding to the vertices $v_{i_{1}},...,v_{i_{k}}$
and which has weight $w(e_{j})=ka_{i_{1}...i_{k}}$.

\end{proof}

\subsection{Uniformisation and homogeneisation processes}

A single tensor is always easier to be used than a family of tensors;
the same apply for homogeneous polynomials. Building a single tensor
from different order tensors requires to fill in the ``gaps''; summing
homogeneous polynomials of varying degrees always give a new polynomial:
but, most frequently this polynomial is no more homogeneous. Homogeneisation
techniques for polynomials are well known and require additional variables.

Different homogeneisation process can be envisaged to get a homogeneous
polynomial that represents a single cubic and symmetric tensor by
making different choices on the variables added in the homogeneisation
phase of the polynomial. As a link has been made between the variables
and the vertices of the hypergraph, we want that this link continue
to occur during the homogeneisation of the polynomial as each term
of the reduced polynomial corresponds to a unique hyperedge in the
original hypergraph; the homogenisation process is interpretable in
term of hypergraph uniformisation process of the original hypergraph:
hypergraph uniformisation process and polynomial homogeneisation process
are the two sides of the same coin. 

So far, we have separated the original hypergraph $\mathcal{H}$ in
layers of increasing $k$-uniform hypergraphs $\mathcal{H}_{k}$ such
that 
\[
\mathcal{H}=\bigoplus\limits _{k=1}^{k_{\max}}\mathcal{H}_{k}.
\]
Each $k$-uniform hypergraph can be represented by a symmetric and
cubic tensor. This symmetric and cubic tensor is mapped to a homogeneous
polynomial. The reduced homogeneous polynomial is interpretable, if
we omit the coefficients of each term, as a disjunctive normal form.
Each term of the homogeneous polynomial is a cunjunctive form which
corresponds to simultaneous presence of vertices in a hyperedge: adding
all the layers allows to retrieve the original hypergraph; adding
the different homogeneous polynomials allows to retrieve the disjunctive
normal form associated with the original hypergraph.

In the hypergraph uniformisation process, iterative steps are done
starting with the lower layers to the upper layers of the hypergraph.
In parallel, the polynomial homogeneisation process is the algebraic
justification of the hypergraph uniformisation process. It allows
to retrieve a polynomial attached to the uniform hypergraph built
at each step and hence a tensor. 

\subsubsection{Hypergraph uniformisation process}

We can describe algorithmically the hypergraph uniformisation process:
it transforms the original hypergraph in a uniform hypergraph.

\subsubsection*{Initialisation}

The initialisation requires that each layer hypergraph is associated
to a weighted hypergraph.

To each uniform hypergraph $\mathcal{H}_{k}$, we associate a weighted
hypergraph $\mathcal{H}_{w_{k},k}=\left(V,E_{k},w_{k}\right)$, with:
$\forall e\in E_{k}:w_{k}(e)=c_{k}$, $c_{k}\in\mathbb{R}^{+*}$.

The coefficients $c_{k}$ are technical coefficients that will be
chosen when considering the homogeneisation process and the fullfillment
of the expectations of the e-adjacency tensor. The coefficients $c_{k}$
can be seen as dilatation coefficients only dependent of the layers
of the original hypergraph.

We initialise:

$k:=1$ and $\mathcal{K}_{w}:=\mathcal{H}_{w_{1},1}$

and generate $k_{\max}-1$ distinct vertices $y_{j}$, $j\in\left\llbracket 1,k_{\max}-1\right\rrbracket $
that are not in $V$.

\subsubsection*{Iterative steps}

Each step in the hypergraph uniformisation process includes three
phases: an inflation phase, a merging phase and a concluding phase.

\paragraph*{Inflation phase: }

The inflation phase consists in increasing the cardinality of each
hyperedge obtained from the hypergraph built at the former step to
reach the cardinality of the hyperedges of the second hypergraph used
in the merge phase. 

\begin{defin}

The \textbf{$y$-vertex-augmented hypergraph} of a weighted hypergraph
$\mathcal{H}_{w}=\left(V,E,w\right)$ is the hypergraph $\overline{\mathcal{H}_{\overline{w}}}=\left(\overline{V},\overline{E},\overline{w}\right)$
obtained by the following rules
\begin{itemize}
\item $y\notin V$;
\item $\overline{V}=V\cup\left\{ y\right\} $; 
\item Writing $\phi:\mathcal{P}\left(V\right)\rightarrow\mathcal{P}\left(\overline{V}\right)$
the map such that for $A\in\mathcal{P}\left(V\right):\,$$\phi(A)=A\cup\left\{ y\right\} $,
it holds:
\begin{itemize}
\item $\overline{E}=\left\{ \phi\left(e\right):e\in E\right\} $;
\item $\forall e\in E$, $\overline{w}\left(\phi(e)\right)=w(e)$.
\end{itemize}
\end{itemize}
\end{defin}

\begin{prop}

The vertex-augmented hypergraph of a $k$-uniform hypergraph is a
$k+1$-uniform hypergraph.

\end{prop}

The inflation phase at step $k$ generates from $\mathcal{K}_{w}$
the $y_{k}$-vertex augmented hypergraph $\overline{\mathcal{K}_{\overline{w}}}$.

As $\mathcal{K}_{w}$ is $k$-uniform at step $k$, $\overline{\mathcal{K}_{\overline{w}}}$
is $k+1$-uniform

\paragraph*{Merging phase:}

The merging phase generates the sum of two weighted hypergraphs called
the merged hypergraph.

\begin{defin}

The \textbf{merged hypergraph} $\widehat{\mathcal{H}_{\widehat{w}}}=\left(\widehat{V},\widehat{E},\widehat{w}\right)$
of two weighted hypergraphs $\mathcal{H}_{a}=\left(V_{a},E_{a},w_{a}\right)$
and $\mathcal{H}_{b}=\left(V_{b},E_{b},w_{b}\right)$ is the weighted
hypergraph defined as follow:
\begin{itemize}
\item $\widehat{V}=V_{a}\cup V_{b}$
\item $\widehat{E_{k+1}}=E_{a}\cup E_{b}$
\item $\forall e\in E_{a}:\,\widehat{w}(e)=\overline{w_{a}}(e)$ and $\forall e\in E_{b}:\,\widehat{w}(e)=w_{b}$
\end{itemize}
\end{defin}

The merging phase at step $k$ generates from $\overline{\mathcal{K}_{w}}$
and $\mathcal{H}_{d,k+1}$ the merged hypergraph $\widehat{\mathcal{K}_{\widehat{w}}}$.
As it is generated from two $k+1$-uniform hypergraph it is also a
$k+1$-uniform hypergraph.

\paragraph*{Step ending phase:}

If $k$ equals $k_{\text{max}}-1$ the iterative part ends up and
return $\widehat{\mathcal{K}_{\widehat{w}}}$. 

Otherwise a next step is need with $\mathcal{K}_{w}:=\widehat{\mathcal{K}_{\widehat{w}}}$
and $k:=k+1$.

\subsubsection*{Termination:}

We obtain by this algorithm a weighted $k_{\max}$-uniform hypergraph
associated to $\mathcal{H}$ which is the returned hypergraph from
the iterative part: we write it $\widehat{\mathcal{H}_{\widehat{w}}}=\left(\widehat{V},\widehat{E},\widehat{w}\right)$.

\begin{defin}

Writing $V_{s}=\left\{ y_{j}:j\in\left\llbracket 1,k_{\max}-1\right\rrbracket \right\} $

$\widehat{\mathcal{H}_{\widehat{w}}}=\left(\widehat{V},\widehat{E},\widehat{w}\right)$
is called the \textbf{$V_{s}$-layered unifom }of $\mathcal{H}$.

\end{defin}

\begin{prop}

Let $\mathcal{H}=\left(V,E\right)$ be a hypergraph of order $k_{\max}$

Let consider $V_{s}=\left\{ y_{j}:j\in\left\llbracket 1,k_{\max}-1\right\rrbracket \right\} $
such that $V\cap V_{s}=\emptyset$ and let $\widehat{\mathcal{H}_{\widehat{w}}}=\left(\widehat{V},\widehat{E},\widehat{w}\right)$
be the \textbf{$V_{s}$-layered unifom }of $\mathcal{H}$. Then:
\begin{itemize}
\item $\left(V,V_{s}\right)$ is a partition of $\widehat{V}$.
\item $\forall e\in E,\exists!\widehat{e}\in\widehat{E}:e\subseteq\widehat{e}\land\widehat{e}\backslash e=\left\{ y_{j}:j\in\left\llbracket \left|e\right|,k_{\max}-1\right\rrbracket \right\} .$
\end{itemize}
\end{prop}

\begin{proof} The way the $V_{s}$-layered uniform of $\mathcal{H}$
is generated justifies the results.

\end{proof}

\begin{prop}

Let $\mathcal{H}=\left(V,E\right)$ be a hypergraph of order $k_{\max}$

Let consider $V_{s,i}=\left\{ y_{j}:j\in\left\llbracket i,k_{\max}-1\right\rrbracket \right\} $
such that $V_{s}=\bigcup\limits _{i\in\left\llbracket 1,k_{\max}-1\right\rrbracket }V_{s,i}$,
$V\cap V_{s}=\emptyset$ and let $\widehat{\mathcal{H}_{\widehat{w}}}=\left(\widehat{V},\widehat{E},\widehat{w}\right)$
be the \textbf{$V_{s}$-layered unifom }of $\mathcal{H}$.

Then:

Vertices of $\mathcal{H}$ that are e-adjacent in $\mathcal{H}$ in
an hyperedge $e$ are e-adjacent with the vertices of $V_{s,\left|e\right|}$
in $\widehat{\mathcal{H}_{\widehat{w}}}$.

Reciprocally, if vertices are e-adjacent in $\widehat{\mathcal{H}_{\widehat{w}}}$,
the ones that are not in $V_{s}$ are e-adjacent in $\mathcal{H}$.

\end{prop}

As a consequence, $\widehat{\mathcal{H}_{\widehat{w}}}$ captures
the e-adjacency of $\mathcal{H}$.

\subsubsection{Polynomial homogeneisation process}

In the polynomial homogeneisation process, we build a new family $R_{\mathcal{H}}=\left(R_{k}\right)$
of homogeneous polynomials of degree $k$ iteratively from the family
of homogeneous polynomials $P_{\mathcal{H}}=\left(P_{k}\right)$ by
following the subsequent steps that respect the phases of construction
in Figure \ref{Fig: Phases of the construction}. Each of these steps
can be linked to the steps of the homogeneisation process. 

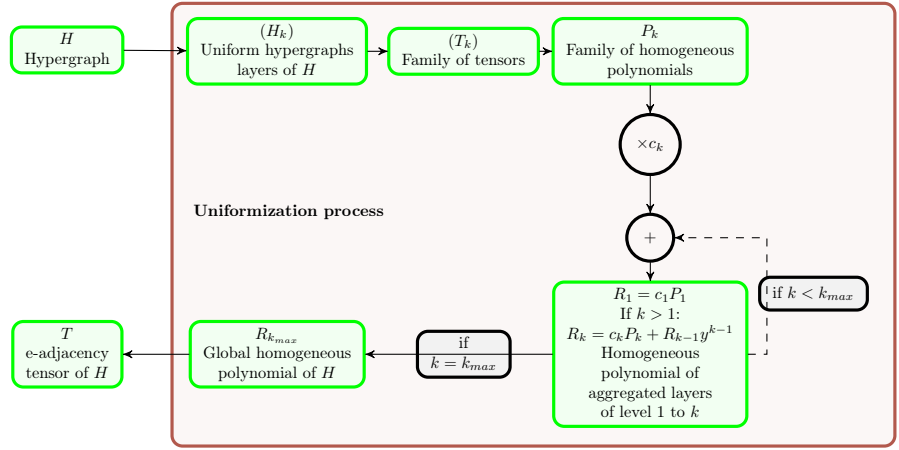
\begin{figure}
\begin{center}
\begin{tikzpicture}[->,>=stealth',scale=0.62, every node/.append style={transform shape}]
\node[state=green] (H) {\begin{tabular}{c}$H$\\Hypergraph\end{tabular}};
\node[state=orange!70!blue,
		right of=H,
		yshift=-3.75cm,
		node distance=10cm,
		minimum width=15.5cm,
		minimum height=9.5cm,
		anchor=center] (Unif) {};
\node[] (Unif_Title) at ([yshift=-4.5cm,xshift=-5.25cm]Unif.north) {\textbf{Uniformization process}};

\node[state=green] (Hk) at  ([yshift=-1.05cm,xshift=-5.5cm]Unif.north){\begin{tabular}{c}$\left( H_k \right)$\\Uniform hypergraphs\\layers of $H$\end{tabular}};
\node[state=green,
		right of=Hk,
		node distance=4cm,
		anchor=center] (Tk) {
			\begin{tabular}{c}
				$\left( T_k \right)$ \\
				Family of tensors
			\end{tabular}
		};
\node[state=green,
		right of=Tk,
		node distance=4cm,
		anchor=center] (Pk) {
			\begin{tabular}{c}
				$P_k$ \\
				Family of homogeneous\\
				polynomials
			\end{tabular}
		};
\node[coeff,
		yshift=-2cm,
		right of=Pk,
		node distance=0cm,
		anchor=center] (ck) {
			\begin{tabular}{c}
				$\times c_k$
			\end{tabular}
		};
\node[coeff,
		yshift=-2cm,
		right of=ck,
		node distance=0cm,
		anchor=center] (plus) {
			\begin{tabular}{c}
				$+$
			\end{tabular}		
		};
\node[state=black,
		yshift=-1.2cm,
		right of=plus,
		node distance=3.5cm,
		anchor=center] (k) {
			\begin{tabular}{c}
				if $k<k_{max}$
			\end{tabular}		
		};

\node[state=green,
		yshift=-2.5cm,
		right of=plus,
		node distance=0cm,
		anchor=center] (Rk) {
			\begin{tabular}{c}
				$R_1=c_1 P_1$\\
				If $k>1$:\\
				$R_k=c_k P_k+R_{k-1} y^{k-1}$\\
				Homogeneous\\
				polynomial of\\
				aggregated layers\\
				of level 1 to $k$\\
			\end{tabular}
		};
\node[state=green,
		yshift=0cm,
		left of=Rk,
		node distance=8cm,
		anchor=center] (Rkmax) {
			\begin{tabular}{c}
				$R_{k_{max}}$\\
				Global homogeneous\\
				polynomial of $H$
			\end{tabular}
		};
\node[state=black,
		yshift=0cm,
		left of=Rk,
		node distance=4cm,
		anchor=center] (kmax) {
			\begin{tabular}{c}
				if\\
				$k=k_{max}$
			\end{tabular}		
		};

\node[state=green,
		yshift=0cm,
		left of=Rkmax,
		node distance=4.5cm,
		anchor=center] (T) {
			\begin{tabular}{c}
				$T$\\
				e-adjacency\\
				tensor of $H$
			\end{tabular}
		};
\path (H) edge (Hk)
(Hk) edge (Tk)
(Tk) edge (Pk)
(Pk) edge (ck)
(ck) edge (plus)
(plus) edge (Rk)
(Rk) edge (Rkmax)
(Rkmax) edge (T);

\draw[->, dashed] (Rk) --++ (2.5,0) --++(0,2.5) --++ (plus);

\end{tikzpicture}
\end{center}

\caption{Different phases of the construction of the e-adjacency tensor}
\label{Fig: Phases of the construction}
\end{figure}

\subsubsection*{Initialisation}

Each polynomial $P_{k}$, $k\in\left\llbracket 1,k_{\max}\right\rrbracket $
attached to the corresponding layer $k$-uniform hypergraph $\mathcal{H}_{k}$
is multiplied by a coefficient $c_{k}$ equals to the dilatation coefficients
of the hypergraph uniformisation process. $c_{k}P_{k}$ represents
the reduced homogeneous polynomial attached to $\mathcal{H}_{w_{k},k}$.

We initialise:

$k:=1$ and $R_{k}\left(\boldsymbol{z}_{k-1}\right)=R_{1}\left(\boldsymbol{z}_{o}\right)=c_{1}P_{1}\left(\boldsymbol{z}_{o}\right)=c_{1}\sum\limits _{i=1}^{n}a_{(1)\,i}z^{i}.$

We generate $k_{\max}-1$ distinct 2 by 2 variables $y^{j}$, $j\in\left\llbracket 1,k_{\max}-1\right\rrbracket $
that are also distinct 2 by 2 from the $z^{i}$, $i\in\left\llbracket 1,n\right\rrbracket $.

\subsubsection*{Iterative steps}

At each step, we sum the current $R_{k}\left(z_{k-1}\right)$ with
the next layer coefficiented polynomial $c_{k+1}P_{k+1}$ in a way
to obtain a homogeneous polynomial $R_{k+1}\left(z_{k}\right)$. To
help the understanding we describe the first step, then generalise
to any step.

\uline{Case \mbox{$k=1$}:} To build $R_{2}$ an homogeneization
of the sum of $R_{1}$ and $c_{2}P_{2}$ is needed. It holds:

\[
R_{1}\left(\boldsymbol{z}_{o}\right)+c_{2}P_{2}\left(\boldsymbol{z}_{o}\right)=c_{1}\sum\limits _{i=1}^{n}a_{(1)\,i}z^{i}+c_{2}\sum\limits _{i_{1},i_{2}=1}^{n}a_{(2)\,i_{1}i_{2}}z^{i_{1}}z^{i_{2}}
\]

To achieve the homogeneization of $R_{1}\left(\boldsymbol{z}_{o}\right)+c_{2}P_{2}\left(\boldsymbol{z}_{o}\right)$
a new variable $y^{1}$ is introduced.

It follows for $y^{1}\neq0$:

\begin{eqnarray*}
R_{2}\left(\boldsymbol{z}_{1}\right) & = & R_{2}\left(z_{0},y^{1}\right)\\
 & = & y^{1(2)}\left(R_{1}\left(\dfrac{z_{0}}{y^{1}}\right)+c_{2}P_{2}\left(\dfrac{z_{0}}{y^{1\,(2)}}\right)\right)\\
 & = & c_{1}\sum\limits _{i=1}^{n}a_{(1)\,i}z^{i}y^{1}+c_{2}\sum\limits _{i_{1},i_{2}=1}^{n}a_{(2)\,i_{1}i_{2}}z^{i_{1}}z^{i_{2}}.
\end{eqnarray*}

By continuous prolongation of $R_{2}$, it is set: 
\[
R_{2}\left(\boldsymbol{z}_{o},0\right)=c_{2}\sum\limits _{i_{1},i_{2}=1}^{n}a_{(2)\,i_{1}i_{2}}z^{i_{1}}z^{i_{2}}.
\]

In this step, the degree 1 coefficiented polynomial $R_{1}\left(z_{0}\right)=c_{1}P_{1}\left(z_{0}\right)$
attached to $\mathcal{H}_{w_{1},1}$ is transformed in a degree 2
homogeneous polynomial $y^{1}R_{1}\left(z_{0}\right)=c_{1}y^{1}P_{1}\left(z_{0}\right)$:
$y^{1}R_{1}\left(z_{0}\right)$ corresponds to the homogeneous polynomial
of the weighted $y_{1}$-vertex-augmented 1-uniform hypergraph $\overline{\mathcal{H}_{\overline{w_{1}},1}}$
built during the inflation phase in the hypergraph uniformisation
process. 

$y_{1}R_{1}\left(z_{0}\right)$ is then summed with the homogeneous
polynomial $c_{2}P_{2}$ attached to $\mathcal{H}_{w_{2},2}$ to get
an homogeneous polynomial of degree 2: $R_{2}\left(z_{1}\right)$.
$R_{2}\left(z_{1}\right)$ is the homogeneous polynomial of the merged
2-uniform hypergraph $\widehat{\mathcal{H}_{\widehat{w_{1}},1}}$
of $\overline{\mathcal{H}_{w_{1},1}}$ and $\mathcal{H}_{w_{2},2}$.

\uline{General case:} Supposing that $R_{k}\left(\boldsymbol{z}_{k-1}\right)$
is an homogeneous polynomial of degree $k$ that can be written as: 

\[
R_{k}\left(\boldsymbol{z}_{k-1}\right)=\sum\limits _{j=1}^{k}c_{j}\sum\limits _{i_{1},...,i_{j}=1}^{n}a_{(j)\,i_{1}...i_{j}}z^{i_{1}}...z^{i_{j}}\prod\limits _{l=j}^{k-1}y^{l},
\]

with the convention that: $\prod\limits _{l=j}^{k-1}y^{l}=1$ if $j>k-1$
and $\omega_{k-1}=z^{1},...,z^{n},y^{1},...,y^{k-1}$

$R_{k+1}$ is built as an homogeneous polynomial from the sum of $R_{k}$
and $c_{k+1}P_{k+1}$ by adding a variable $y_{k}$ and factorizing
by its $k+1$-th power. 

Therefore, for $y^{k-1}\neq0$:

\begin{eqnarray*}
R_{k+1}\left(\boldsymbol{z}_{k}\right) & = & y^{k\,(k+1)}\left(R_{k}\left(\dfrac{\boldsymbol{z}_{k-1}}{y^{k\,(k)}}\right)+c_{k+1}P_{k+1}\left(\dfrac{\boldsymbol{z}_{o}}{y^{k\,(k+1)}}\right)\right)\\
 & = & \left(\sum\limits _{j=1}^{k}c_{j}\sum\limits _{i_{1},...,i_{j}=1}^{n}a_{(j)\,i_{1}...i_{j}}z^{i_{1}}...z^{i_{j}}\prod\limits _{l=j}^{k-1}y^{l}\right)y^{k}\\
 &  & +c_{k+1}\sum\limits _{i_{1},...,i_{k+1}=1}^{n}a_{(k+1)\,i_{1}\,...\,i_{k+1}}z^{i_{1}}...z^{i_{k+1}}
\end{eqnarray*}

And for $y_{k}=0$, it is set by continuous prolongation: $R_{k+1}\left(\boldsymbol{z}_{k-1},0\right)=c_{k+1}\sum\limits _{i_{1},...,i_{k+1}=1}^{n}a_{(k+1)\,i_{1}\,...\,i_{k+1}}z^{i_{1}}...z^{i_{k+1}}.$

The fact that $P_{k+1}\left(z_{0}\right)$ can be null doesn't prevent
to do the step: the degree of $R_{k}$ will then be elevated of 1.

The interpretation of this step is similar to the one done for the
case $k=1$.

\paragraph*{Step ending phase:}

If $k$ equals $k_{\text{max}}-1$ the iterative part ends up, else
$k:=k+1$ and the next iteration is started.

\subsubsection*{Conclusion}

The algorithm build a family of homogeneous polynomial which is interpretable
in term of uniformisation of a hypergraph.

\subsection{Building an unnormalized symmetric tensor from this family of homogeneous
polynomials}

\textbf{Based on $R_{\mathcal{H}}$}

It is now valuable to interpret the built polynomials. 

The notation $\boldsymbol{w}_{(k)}=w_{(k)}^{1},...,w_{(k)}^{n+k-1}$
is used.
\begin{itemize}
\item The interpretation of $R_{1}$ is trivial as it holds the single element
hyperedges of the hypergraph.
\item $R_{2}$ is an homogeneous polynomial with $n+1$ variables of order
2.

\begin{eqnarray*}
R_{2}\left(\boldsymbol{z}_{1}\right) & = & c_{1}\sum\limits _{i=1}^{n}a_{(1)\,i}z^{i}y^{1}+c_{2}\sum\limits _{i_{1},i_{2}=1}^{n}a_{(2)\,i_{1}i_{2}}z^{i_{1}}z^{i_{2}}\\
 & = & c_{1}\sum\limits _{1\leqslant i\leqslant n}\alpha_{(1)\,i}z^{i}y^{1}+c_{2}\sum\limits _{1\leqslant i_{1}\leqslant i_{2}\leqslant n}\alpha_{(2)\,i_{1}i_{2}}z^{i_{1}}z^{i_{2}}
\end{eqnarray*}
It can be rewritten:
\[
R_{2}\left(\boldsymbol{w}_{(2)}\right)=\sum\limits _{i_{1},i_{2}=1}^{n+1}r_{(2)\,i_{1}i_{2}}w_{(2)}^{i_{1}}w_{(2)}^{i_{2}}
\]
where:
\begin{itemize}
\item for $1\leqslant i\leqslant n$: $w_{(2)}^{i}=z^{i}$ 
\item $w_{(2)}^{n+1}=y^{1}$ 
\item for $1\leqslant i_{1}\leqslant i_{2}\leqslant n$ and $\sigma\in\mathcal{S}_{2}$:
\[
r_{(2)\,\sigma\left(i_{1}\right)\sigma\left(i_{2}\right)}=\dfrac{c_{2}\alpha_{(2)\,i_{1}i_{2}}}{2!}=c_{2}a_{(2)\,i_{1}i_{2}}
\]
\item for $1\leqslant i\leqslant n$ and $\sigma\in\mathcal{S}_{2}$: 
\[
r_{(2)\,\sigma(i)\,\sigma(n+1)}=\dfrac{c_{1}\alpha_{(1)\,i}}{2!}=\dfrac{c_{1}a_{(1)\,i}}{2!}
\]
\item the other coefficients: $r_{(2)\,i_{1}i_{2}}$ are null.
\end{itemize}
Also $R_{2}$ can be linked to a symmetric hypercubic tensor of order
2 and dimension $n+1$.
\item $R_{k}$ is an homogeneous polynomial with $n+k-1$ variables of order
$k$. 

\begin{eqnarray*}
R_{k}\left(\boldsymbol{z}_{k-1}\right) & = & \sum\limits _{j=1}^{k}c_{j}\sum\limits _{i_{1},...,i_{j}=1}^{n}a_{(j)\,i_{1}...i_{j}}z^{i_{1}}...z^{i_{j}}\prod\limits _{l=j}^{k-1}y^{l}\\
 & = & \sum\limits _{j=1}^{k}c_{j}\sum\limits _{1\leqslant i_{1}\leqslant...\leqslant i_{j}\leqslant n}\alpha_{(j)\,i_{1}...i_{j}}z^{i_{1}}...z^{i_{j}}\prod\limits _{l=j}^{k-1}y^{l}
\end{eqnarray*}
with the convention that: $\prod\limits _{l=j}^{k-1}y^{l}=1$ if $j>k-1$.

It can be rewritten:
\[
R_{k}\left(\boldsymbol{w}_{(k)}\right)=\sum\limits _{i_{1},...,i_{k}=1}^{n+k-1}r_{(k)\,i_{1}\,...\,i_{k}}w_{(k)}^{i_{1}}...w_{(k)}^{i_{k}}
\]
where:
\begin{itemize}
\item for $1\leqslant i\leqslant n$: $w_{(k)}^{i}=z^{i}$
\item for $n+1\leqslant i\leqslant n+k-1$: $w_{(k)}^{i}=y^{i-n}$
\item for $1\leqslant i_{1}\leqslant...\leqslant i_{k}\leqslant n$, for
all $1\leqslant j\leqslant k-1$, for all $\sigma\in\mathcal{S}_{k}$: 
\begin{itemize}
\item $r_{(k)\,\sigma\left(i_{1}\right)...\sigma\left(i_{k}\right)}=\dfrac{c_{k}\alpha_{(k)\,i_{1}...i_{k}}}{k!}=c_{k}a_{(k)\,i_{1}...i_{k}}$ 
\item $r_{(k)\,\sigma\left(i_{1}\right)...\sigma\left(i_{j}\right)\sigma(n+j)...\sigma(n+k-1)}=\dfrac{c_{j}\alpha_{(j)\,i_{1}...i_{j}}}{k!}=\dfrac{j!}{k!}c_{j}a_{(j)\,i_{1}...i_{j}}$
\end{itemize}
\item the other elements $r_{(k)\,i_{1}\,...\,i_{k}}$ are null.
\end{itemize}
Also $R_{k}$ can be linked to a symmetric hypercubic tensor of order
$k$ and dimension $n+k-1$ written $\boldsymbol{R_{k}}$ whose elements
are $r_{(k)\,i_{1}\,...\,i_{k}}$.
\end{itemize}
The hypermatrix $\boldsymbol{R_{k_{\max}}}$ is called the unnormalized
tensor.

\subsection{Interpretation and choice of the coefficients for the unnormalized
tensor}

There are different ways of setting the coefficients $c_{1},...,c_{k_{\text{\text{max}}}}$
that are used. These coefficients can be seen as a way of normalizing
the tensors of e-adjacency generated from the $k$-uniform hypergraphs.

A first way of choosing them is to set them all equal to 1. In this
case no normalization occurs. The impact on the e-adjacency tensor
of the original hypergraph is that e-adjacency in hyperedges of size
$k$ have a weight of $k$ times bigger than the e-adjacency in hyperedges
of size 1.

A second way of choosing these coefficients is to consider that in
a $k$-uniform hypergraph, each hyperedge holds $k$ vertices and
then contributes to $k$ to the total degree. Representing this $k$-uniform
hypergraph by the $k-$adjacency degree normalized tensor $\mathcal{A}_{k}=\left(a_{(k)\,i_{1}...i_{k}}\right)_{1\leqslant i_{1},...,i_{k}\leqslant n}$,
it holds a revisited hand-shake lemma for $k$-uniform hypergraphs:
\begin{eqnarray*}
\sum\limits _{1\leqslant i_{1},...,i_{k}\leqslant n}a_{(k)\,i_{1}...i_{k}} & = & \sum\limits _{i=1}^{n}\sum\limits _{1\leqslant i_{2},...,i_{k}\leqslant n}a_{(k)\,ii_{2}...i_{k}}\\
 & = & \sum\limits _{i=1}^{n}d_{(k)\,i}\\
 & = & k\left|E_{k}\right|
\end{eqnarray*}

where $d_{(k)\,i}$ is the degree of the vertex $v_{i}$ in $\mathcal{H}_{k}$.

This formula can be extended to general hypergraphs: 
\begin{eqnarray*}
\left|E\right| & = & \sum\limits _{k=1}^{k_{\text{\text{max}}}}\left|E_{k}\right|\\
 & = & \sum\limits _{k=1}^{k_{\text{\text{max}}}}\dfrac{1}{k}\sum\limits _{i=1}^{n}d_{(k)\,i}\\
 & = & \sum\limits _{k=1}^{k_{\text{\text{max}}}}\dfrac{1}{k}\sum\limits _{1\leqslant i_{1},...,i_{k}\leqslant n}a_{(k)\,i_{1}...i_{k}}.
\end{eqnarray*}

For general hypergraphs, the tensor is of order $k_{\max}$. 
\begin{eqnarray*}
\sum\limits _{1\leqslant i_{1},...,i_{k_{\max}}\leqslant n+k_{\max}-1}r_{i_{1}...i_{k_{\max}}} & = & \sum\limits _{i=1}^{n+k_{\max}-1}\sum\limits _{1\leqslant i_{2},...,i_{k}\leqslant n+k_{\max}-1}r_{ii_{2}...i_{k_{\max}}}\\
 & = & \sum\limits _{i=1}^{n}\deg\left(v_{i}\right)+\sum\limits _{i=n+1}^{n+k_{\max}-1}\deg\left(y_{i}\right)
\end{eqnarray*}

The constructed tensor corresponds to the tensor of a $k_{\max}$-uniform
hypergraph with $n+k_{\max}-1$ vertices. It holds:
\[
\sum\limits _{1\leqslant i_{1},...,i_{k_{\max}}\leqslant n+k_{\max}-1}r_{i_{1}...i_{k_{\max}}}=k_{\max}\text{\ensuremath{\left|E\right|}}.
\]

And therefore: 

\[
\sum\limits _{1\leqslant i_{1},...,i_{k_{\max}}\leqslant n+k_{\max}-1}r_{i_{1}...i_{k_{\max}}}=\sum\limits _{k=1}^{k_{\text{\text{max}}}}\dfrac{k_{\max}}{k}\sum\limits _{1\leqslant i_{1},...,i_{k}\leqslant n}a_{(k)\,i_{1}...i_{k}}.
\]

Also $c_{k}=\dfrac{k_{\max}}{k}$ seems to be a good choice in this
case.

The final choice will be taken in the next paragraph to answer to
the required specifications on degrees. It will also fix the matrix
chosen for the uniform hypergraphs.

\subsection{Unnormalized e-adjacency tensor's expectations fulfillment}

\label{subsec:Gages-of-the}

\begin{gage}The tensor should be symmetric and its generation should
be simple.\end{gage}

\begin{proof}By construction the e-adjacency tensor is symmetric.
To generate it only one element has to be described for a given hyperedge
the other elements obtained by permutation of the indices being the
same. Also the built e-adjacency tensor is fully described by giving
$\left|E\right|$ elements.

\end{proof} 

\begin{gage}The unnormalized e-adjacency tensor keeps the overall
structure of the hypergraph.\end{gage}

\begin{proof}It is inherent to the way the tensor has been built:
the layer of level equal or under $j$ can be seen in the mode 1 at
the $n+j$-th component of the mode. To have only elements of level
$j$ one can project this mode so that it keeps only the first $n$
dimensions.\end{proof}

In the expectations of the built co-tensors listed in the paragraph
\ref{subsec:Expectations-for-a}, the e-adjacency tensor should allow
the retrieval of the degree of the vertices. It implies to fix the
choice of the $k$-adjacency tensors used to model each layer of the
hypergraph as well as the normalizing coefficient.

Let consider for $1\leqslant k\leqslant k_{\max}$, $2\leqslant l\leqslant k_{\max}$
and $1\leqslant i\leqslant n+k_{\max}-1$: 
\[
I_{k,l,i}=\left\{ \left(i_{1},...,i_{l}\right):i_{1}=i\land\forall j\in\left\llbracket 2,l\right\rrbracket :1\leqslant i_{j}\leqslant n+k-1\right\} 
\]
 and its subset of ordered tuples 
\[
OI_{k,l,i}=\left\{ \left(i_{1},...,i_{l}\right):\left(i_{1},...,i_{l}\right)\in I_{k,l,i}\land\left(l\geqslant2\implies\forall\left(j_{1},j_{2}\right)\in\left\llbracket 2,l\right\rrbracket ^{2}:j_{1}<j_{2}\implies i_{j_{1}}<i_{j_{2}}\right)\right\} .
\]

Then: 
\begin{eqnarray*}
\sum\limits _{\left(i_{1},...,i_{k_{\text{\text{max}}}}\right)\in I_{k_{\text{\text{max}}},k_{\text{\text{max}}},i}}r_{i_{1}...i_{k_{\text{max}}}} & = & \sum\limits _{\left(i_{1},...,i_{k_{\text{\text{max}}}}\right)\in OI_{k_{\text{\text{max}}},k_{\text{\text{max}}},i}}\left(k_{\text{\text{max}}}-1\right)!r_{i_{1}...i_{k_{\text{\text{max}}}}}\\
 & = & \sum\limits _{j=1}^{k_{\text{\text{max}}}}\sum\limits _{\left(i_{1},...,i_{j}\right)\in OI_{k_{\text{\text{max}}},j,i}}\dfrac{j!c_{j}a_{(j)\,i_{1}...i_{j}}}{k_{\max}}
\end{eqnarray*}

Hence, the expectation on the retrieval of degree imposes to set $c_{j}a_{(j)\,i_{1}...i_{j}}=\dfrac{k_{\max}}{j!}$
for the elements of $\mathcal{A}_{(j)}$ that are not null, which
is coherent with the usage of the coefficient $c_{j}=\dfrac{k_{\max}}{j}$
and of the degree-normalized tensor for $j$-uniform hypergraph where
not null elements are equals to: $\dfrac{1}{(j-1)!}$. This choice
is then made for the rest of the article.

\begin{rmk}

By choosing $c_{j}=\dfrac{k_{\max}}{j}$ and the degree-normalized
tensor for $j$-uniform hypergraph where not null elements are equals
to: $\dfrac{1}{(j-1)!}$, it follows that: $r_{i_{1}...i_{k_{\text{\text{max}}}}}=\dfrac{1}{\left(k_{\text{\text{max}}}-1\right)!}$
for all elements which is consistent with the fact that we have built
a $k_{\max}$-uniform hypergraph by filling each hyperedge with additional
vertices. This method is similar to make a plaster molding from a
footprint in the sand: the filling elements help reveal the structure
behind.

\end{rmk}

With this choice, writing $\mathds{1}_{e\in E}:\begin{cases}
1 & \text{if }e\in E\\
0 & \text{otherwise}
\end{cases}.$

\begin{eqnarray*}
\sum\limits _{\left(i_{1},...,i_{k_{\text{\text{max}}}}\right)\in I_{k_{\text{\text{max}}},k_{\text{\text{max}}},i}}r_{i_{1}...i_{k_{\text{\text{max}}}}} & = & \sum\limits _{j=1}^{k_{\text{\text{max}}}}\sum\limits _{\left(i_{1},...,i_{j}\right)\in OI_{k_{\text{\text{max}}},j,i}}\mathds{1}_{\left\{ v_{i_{1}},...,v_{i_{j}}\right\} \in E}.
\end{eqnarray*}

It follows immediately:

\begin{gage}The unnormalized e-adjacency tensor allows the retrieval
of the degree of the vertices of the hypergraph.\end{gage}

\begin{proof}Defining for $1\leqslant i\leqslant n$: $d_{i}=\deg\left(v_{i}\right)$.

From the previous choice, it follows that: 
\begin{eqnarray*}
\sum\limits _{\left(i_{1},...,i_{k_{\text{\text{max}}}}\right)\in I_{k_{\text{\text{max}}},k_{\text{\text{max}}},i}}r_{i_{1}...i_{k_{\text{\text{max}}}}} & = & \sum\limits _{\left(i_{1},...,i_{k_{\text{\text{max}}}}\right)\in OI_{k_{\text{\text{max}}},k_{\text{\text{max}}},i}}\left(k_{\text{\text{max}}}-1\right)!r_{i_{1}...i_{k_{\text{\text{max}}}}}\\
 & = & \sum\limits _{j=1}^{k_{\text{\text{max}}}}\sum\limits _{\left(i_{1},...,i_{j}\right)\in OI_{k_{\text{\text{max}}},j,i}}\dfrac{j!c_{j}a_{(j)\,i_{1}...i_{j}}}{k_{\max}}\\
 & = & \sum\limits _{j=1}^{k_{\text{\text{max}}}}\sum\limits _{\left(i_{1},...,i_{j}\right)\in OI_{k_{\text{\text{max}}},j,i}}\mathds{1}_{\left\{ v_{i_{1}},...,v_{i_{j}}\right\} \in E}\\
 & = & \deg\left(v_{i}\right)
\end{eqnarray*}
 as $\dfrac{j!c_{j}a_{(j)\,i_{1}...i_{j}}}{k_{\max}}=1$ only for
hyperedges where $v_{i}$ is in it (and they are counted only once
for each hyperedge).

\end{proof}

\begin{gage}The unnormalized e-adjacency tensor allows the retrieval
of the cardinality of the hyperedges.\end{gage}

\begin{proof}Defining $d_{n+i}=\left|\left\{ e\,:\,\left|e\right|\leqslant i\right\} \right|$
for $1\leqslant i\leqslant k_{\max}$.

\begin{eqnarray*}
\sum\limits _{\left(i_{1},...,i_{k_{\text{\text{max}}}}\right)\in I_{k_{\text{max}},k_{\text{\text{max}}},n+i}}r_{i_{1}...i_{k_{\text{max}}}} & = & \sum\limits _{j=1}^{i}\sum\limits _{1\leqslant l_{1}<...<l_{j}\leqslant n}\left(k_{\text{\text{max}}}-1\right)!\dfrac{j!c_{j}a_{(j)\,i_{1}...i_{j}}}{k_{\max}!}\\
 & = & \sum\limits _{j=1}^{i}\sum\limits _{1\leqslant l_{1}<...<l_{j}\leqslant n}\mathds{1}_{\left\{ v_{l_{1}},...,v_{l_{j}}\right\} \in E}
\end{eqnarray*}

due to the fact that $r_{n+i\,i_{2}...i_{k_{\text{max}}}}\neq0$ if
and only if it exists at most $i$ indices $i_{2}$ to $i_{k_{\max}}$
that are between 1 and $n$ which correspond to vertices in the general
hypergraph and the other indices have value strictly above $n$ which
represent additional vertices.

It follows:

\begin{eqnarray*}
\sum\limits _{\left(i_{1},...,i_{k_{\text{\text{max}}}}\right)\in I_{k_{\text{\text{max}}},k_{\text{\text{max}}},n+i}}r_{i_{1}...i_{k_{\text{\text{max}}}}} & = & d_{n+i}.
\end{eqnarray*}

We set: $d_{n+k_{\max}}=\left|E\right|$.

Also $d_{n+j}$ allows to retrieve the number of hyperedges of cardinality
equal or less than $j$.

Therefore: 
\begin{itemize}
\item for $2\leqslant j\leqslant k_{\max}$: $\left|\left\{ e\,:\,\left|e\right|=j\right\} \right|=d_{n+j}-d_{n+j-1}$ 
\item for $j=1$: $\left|\left\{ e\,:\,\left|e\right|=1\right\} \right|=d_{n+1}$
\end{itemize}
An other way of keeping directly the cardinality of the layer $k_{\max}$
in the e-adjacency tensor would be to store it in an additional variable
$y_{k_{\max}}$.

\end{proof}

\begin{gage}The e-adjacency tensor is unique up to the labeling of
the vertices for a given hypergraph.

Reciprocally, given the e-adjacency tensor and the number of vertices,
the associated hypergraph is unique.

\end{gage}

\begin{proof}

Given a hypergraph, the process of decomposition in layers is bijective
as well as the formalization by degree normalized $k$-adjacency tensor.
Given the coefficients, the process of building the e-adjacency homogeneous
polynomial is also unique and the reversion to a symmetric cubic tensor
is unique.

Given the e-adjacency tensor and the number of vertices, as the e-adjacency
tensor is symmetric, up to the labeling of the vertices, considering
that the first $n$ variables encoded in the e-adjacency tensor in
each direction represents variables associated to vertices of the
hypergraph and the last variables in each direction encode the information
of cardinality. Therefore it is possible to retrieve each layer of
the hypergraph uniquely and consequently the whole hypergraph.

\end{proof}

\subsection{Interpretation of the e-adjacency tensor}

The general hypergraph layer decomposition allows to retrieve uniform
hypergraphs that can be separately modeled by e-adjacency (or equivalently
$k$-adjacency) tensor of $k$-uniform hypergraphs. We have shown
that filling these different layers with additional vertices allow
to uniformize the original hypergraph by keeping the e-adjacency.
The coefficients used in the iterative process has to be seen as weights
on the hyperedges of the final $k_{\max}$-uniform hypergraph: these
coefficients allow to retrieve the right number of edges from the
uniformized hypergraph tensor so that it corresponds to the number
of edges of the original hypergraph.

The additional dimensions in the e-adjacency tensor allows to retrieve
the cardinality of the hyperedges. By decomposing a hypergraph in
a set of uniform hypergraphs the hyperedges are quotiented depending
on their cardinality.

The iterative approach principle is illustrated in Figure\ref{Fig: Iterative approach}:
vertices that are added at each level give indication on the original
cardinality of the hyperedge it is added to.

\begin{figure}
\begin{center}\includegraphics[scale=0.3]{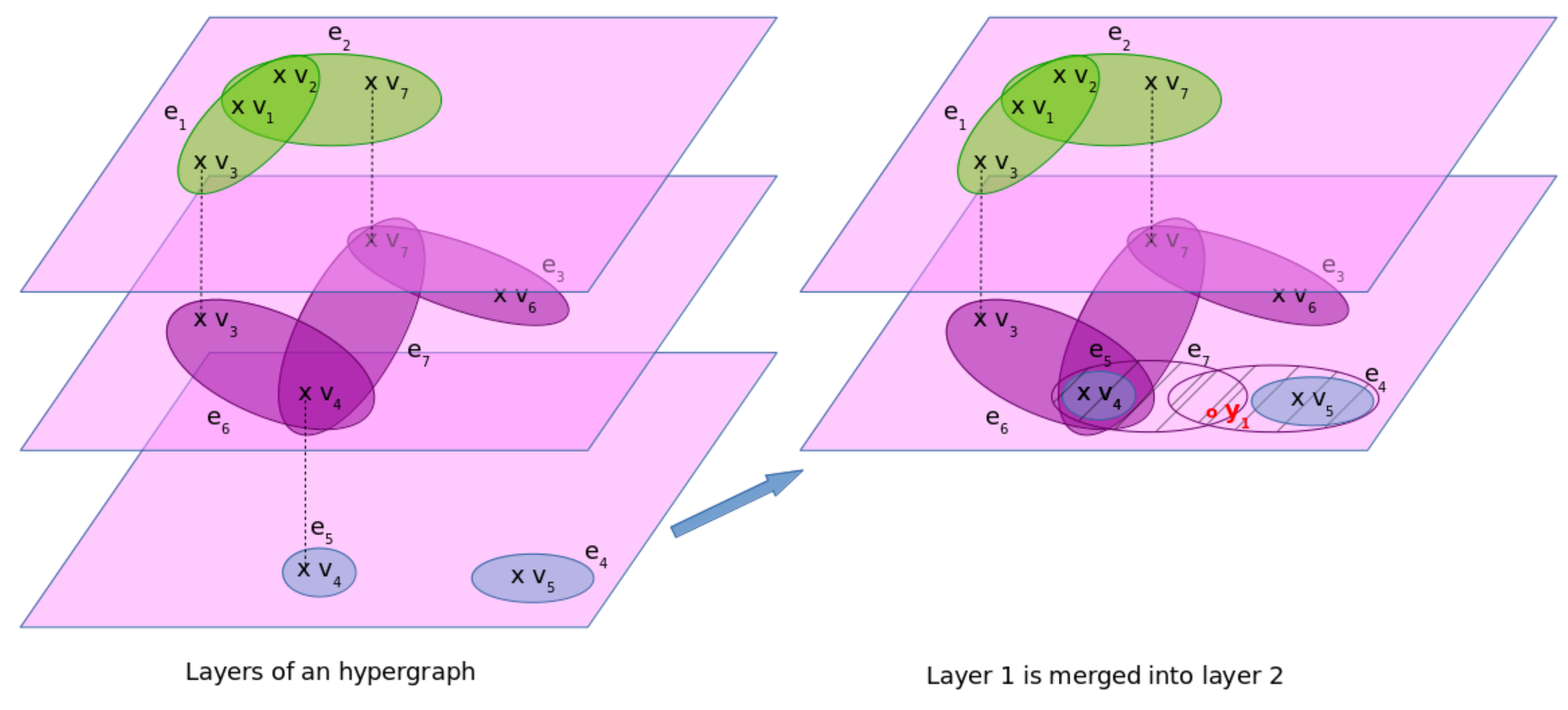}\end{center}

\begin{center}\includegraphics[scale=0.3]{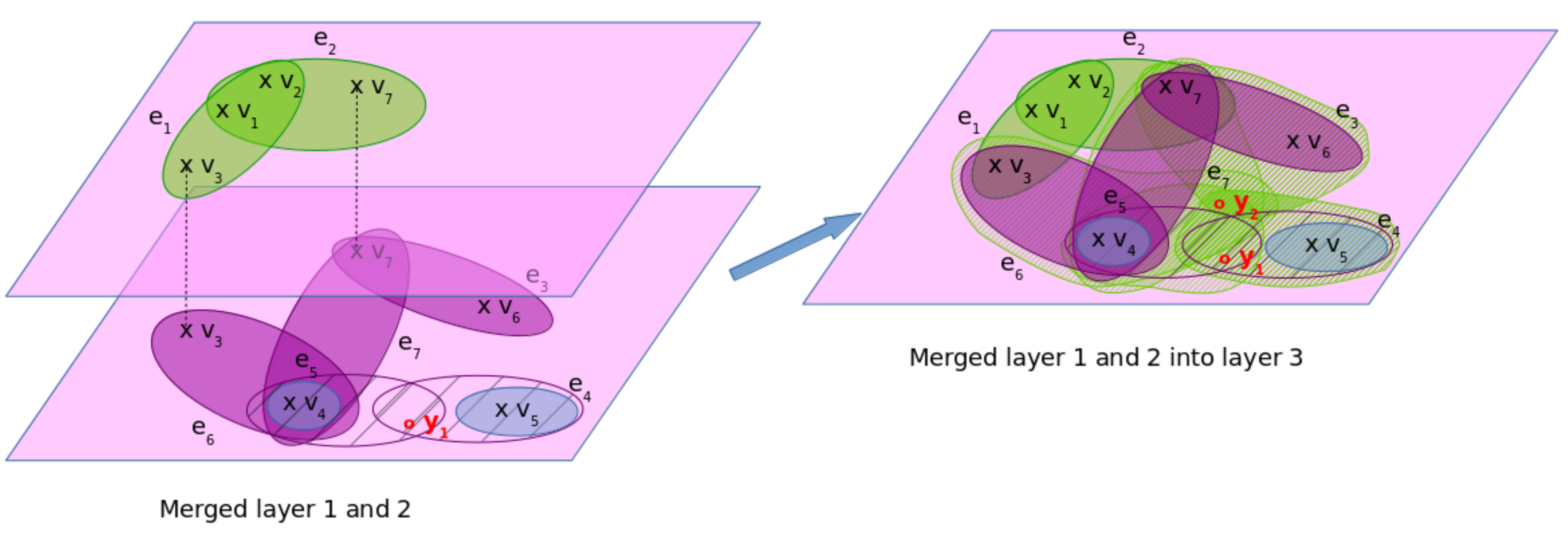}\end{center}

\caption{Illustration of the iterative approach concept on an example}

In the iterative approach the layers of level $n$ and $n+1$ are
merged together into the layer $n+1$ by adding a filling vertex to
the hyperedges of the layer $n$. On this example, during the first
step the layer 1 and 2 are merged to form a 2-uniform hypergraph.
In the second step, the 2-uniform hypergraph obtained in the first
step is merged to the layer 3 to obtain a 3-uniform hypergraph.\label{Fig: Iterative approach}
\end{figure}

Viewed in an other way, e-adjacency hypermatrix of uniform hypergraph
don't need an extra dimension as the hyperedges are uniform, therefore
there is no ambiguity. Adding an extra variable allows to capture
the dimensionality of each hyperedge meanwhile preventing any ambiguity
on the meaning of each element of the tensor.

\section{Some comments on the e-adjacency tensor}

\label{sec:Some-comments-on}

\subsection{The particular case of graphs}

\label{sec:The-particular-case}

As a graph $G=\left(V,E\right)$ with $\left|V\right|=n$ can always
be seen a $2$-uniform hypergraph $\mathcal{H}_{G}$, the approach
given in this paragraph should allow to retrieve in a coherent way
the spectral theory for normal graphs.

The hypergraph that contains the 2-uniform hypergraph is then composed
of an empty level 1 layer and a level 2 layer that contains only $\mathcal{H}_{G}$.

Let $A$ be the adjacency matrix of $G$. The e-adjacency tensor of
the corresponding 2-uniform hypergraph is of order 2 and obtained
from $A$ by multiplying it by $c_{2}$ and adding one row and one
column of zero. Therefore the e-adjacency tensor of the two level
of the corresponding hypergraph is: $\mathcal{A}$=
$\left(\begin{array}{c|c}
c_2A & 0\\
\hline
0 & 0
\end{array}\right)$

Also as an eigenvalue $\lambda$ of $\mathcal{A}$ seen as a matrix
is a solution of the characteristic polynomial $\det\left(\mathcal{A}-\lambda I\right)=0\Leftrightarrow-\lambda\det\left(c_{2}A-\lambda I\right)=0\Leftrightarrow-\lambda c_{2}^{n}\det\left(A-\dfrac{\lambda}{c_{2}}I\right)=0$,
the eigenvalues of $\mathcal{A}$ are $c_{2}$times the ones of $A$
and one additional 0 eigenvalue. This last eigenvalue is attached
to the eigenvector $\left(0...0\,1\right)^{T}$. The other eigenvalues
have same eigenvectors than $A$ with one additional $n+1$ component
which is 0.

\begin{proof}

Let consider $Y=\left(\begin{array}{c}
X\\
y
\end{array}\right)$ with $X$ vector of dimension n. Let $\lambda$ be an eigenvalue
of $\mathcal{A}$ and $Y$ an eigenvector of $\mathcal{A}$

$\mathcal{A}Y=\lambda Y\Leftrightarrow\mathcal{A}\left(\begin{array}{c}
X\\
y
\end{array}\right)=\lambda\left(\begin{array}{c}
X\\
y
\end{array}\right)\Leftrightarrow\left(c_{2}A-\lambda I_{n}\right)X=0\land-\lambda y=0\Leftrightarrow X$ eigenvalue of $A$ attached to $\dfrac{\lambda}{c_{2}}$. $y$ can
be always taken equals to 0 to fit the second condition.

\end{proof}

Therefore globally there is no change in the spectra: the eigenvectors
hold, the eigenvalues of the initial graph are multiplied by the normalizing
coefficient.

\subsection{e-adjacency tensor and DNF}

Let $\mathcal{H}=\left(V,E\right)$ be a hypergraph, $\mathcal{A}$
its e-adjacency tensor and $\tilde{R}_{k_{\max}}$ the reduced attached
homogeneous polynomial.

\[
\tilde{R}_{k_{\max}}\left(\boldsymbol{w}_{\left(k_{\max}\right)}\right)=\sum\limits _{1\leqslant i_{1}<...<i_{k}\leqslant n+k_{\max}-1}\tilde{r}_{\left(k_{\max}\right)\,i_{1}\,...\,i_{k_{\max}}}w_{(k_{\max})}^{i_{1}}...w_{(k_{\max})}^{i_{k_{\max}}}
\]

with $\tilde{r}_{\left(k_{\max}\right)\,i_{1}\,...\,i_{k_{\max}}}=k_{\max}!r_{\left(k_{\max}\right)\,i_{1}\,...\,i_{k_{\max}}}$

The variables $\left(w_{(k_{\max})}^{i}\right)_{1\leqslant i\leqslant n+k_{\max}-1}$
of $R_{k_{\text{\text{\text{max}}}}}$ can be considered as boolean
variables and therefore $R_{k_{\text{\text{\text{max}}}}}$ can be
considered as a boolean function. The variables $w_{(k_{\max})}^{i}$
for $1\leqslant i\leqslant n$ captures the belonging of a vertex
to the considered hyperedge and for $n+1\leqslant i\leqslant n+k_{\max}-1$
to the layer of level $i-n$.

This boolean homogeneous polynomial $P_{B}\left(\boldsymbol{w}_{\left(k_{\max}\right)}\right)$
is in full disjunctive normal form as it is a sum of products of boolean
variables holding only once in each product and where the conjunctive
terms are made of $k_{\max}$ variables.

$P_{B\,k_{\max}}\left(z_{0}\right)=P_{B}\left(z_{0},\underset{k_{\max}-1}{\underbrace{0,...,0}}\right)$
allows to retrieve the part of the full DNF which stores hyperedges
of size $k_{\max}$.

$P_{B\,k_{\max}-1}\left(z_{0}\right)=P_{B}\left(z_{0},\underset{k_{\max}-2}{\underbrace{0,...,0}},1\right)-P_{B}\left(z_{0},\underset{k_{\max}-1}{\underbrace{0,...,0}}\right)$
allows to retrieve the full DNF which stores hyperedges of size $k_{\max}-1$.

$P_{B\,k_{\max}-j}\left(z_{0}\right)=P_{B}\left(z_{0},\underset{k_{\max}-j-1}{\underbrace{0,...,0}},\underset{j}{\underbrace{1,...,1}}\right)-P_{B}\left(z_{0},\underset{k_{\max}-j-2}{\underbrace{0,...,0}},\underset{j-1}{\underbrace{1,...,1}}\right)$
allows to retrieve the full DNF which stores hyperedges of size $k_{\max}-j$.

Stopping at $P_{B\,1}\left(z_{0}\right)=P_{B}\left(z_{0},\underset{k_{\max}-1}{\underbrace{1,...,1}}\right)-P_{B}\left(z_{0},0,\underset{k_{\max}-2}{\underbrace{1,...,1}}\right)$
allows to retrieve the full DNF which stores hyperedges of size 1.

Considering the adjacency matrix of \citet{zhou2007learning} of this
unweighted hypergraph, it holds that $\boldsymbol{w}_{0}^{\top}A\boldsymbol{w}_{0}$
can be considered as a boolean homogeneous polynomial in full disjunctive
form where the conjunctive terms are composed of only two variables,
which shows if it was necessary that this approach is a pairwise approximation
of the e-adjacency tensor.

The homogeneous polynomial attached to \citet{banerjee2017spectra}
tensor can be mapped to a boolean polynomial function by considering
the same term elements with coefficient being 1 when the original
homogeneous polynomial has a non-zero coefficient and 0 otherwise.
This boolean function nonetheless is no more in DNF. Reducing it to
DNF yields to the expression of $P_{B}\left(\boldsymbol{z}_{0},\underset{k_{\max}-1}{\underbrace{1,...,1}}\right)$.

\subsection{Some first results on spectral analysis}

\subsubsection{Eigenvalues of tensors}

The definitions and results of this sub-section are based on \citet{qi2017siam}.
Proofs can be consulted in this reference.

Let $T_{m,n}$ be the set of all real tensors of order m and dimension
$n$ and $S_{m,n}$ the subset of $T_{m,n}$ where all tensors are
symmetric, i.e. invariant under a permutation of the indices of its
elements.

Let $\mathcal{A}=\left(a_{i_{1}...i_{m}}\right)\in T_{m,n}$. Let
$\mathcal{I}\in T_{m,n}$ designates the identity tensor. 

\begin{defin}

A number $\lambda\in\mathbb{C}$ is an \textbf{eigenvalue} of $\mathcal{A}$
if it exists a nonzero vector $x\in\mathbb{C}^{n}$ such that: 
\begin{equation}
\forall i\in\left\llbracket 1,n\right\rrbracket ,\,\left(\mathcal{A}x^{m-1}\right)_{i}=\lambda x_{i}^{m-1}\label{eq:eigenvalue}
\end{equation}

In this case $x$ is called an \textbf{eigenvector} of $\mathcal{A}$
associated with the eigenvalue $\lambda$ and $\left(x,\lambda\right)$
is called an \textbf{eigenpair} of $\mathcal{A}$.

The set of all eigenvalues of $\mathcal{A}$ is called the spectrum
of $\mathcal{A}$. The largest modulus of all eigenvalues is called
the spectra radius of $\mathcal{A}$, denoted as $\rho\left(\mathcal{A}\right)$.

\end{defin}

\begin{rmk}

By writing $x^{[m-1]}=\left(x_{i}^{m-1}\right)_{1\leqslant i\leqslant n}$,
\ref{eq:eigenvalue} can be written:

\[
\mathcal{A}x^{m-1}=\lambda x^{[m-1]}.
\]

\end{rmk}

\begin{prop}

Let $\alpha$ and $\beta$ be two real numbers.

If $\left(\lambda,x\right)$ is an eigenpair of $\mathcal{A}$, then
$\left(\alpha\lambda+\beta,x\right)$ is an eigenpair of $\alpha\mathcal{A}+\beta\mathcal{I}$

\end{prop}

\begin{defin}

A H-eigenvalue is an eigenvalue $\lambda$ of $\mathcal{A}$ that
has a real eigenvector $x$ associated to it. $x$ is called in this
case an $H$-eigenvector.

\end{defin}

\begin{prop}

A H-eigenvalue is real.

\end{prop}

A real eigenvalue is not necessarily a H-eigenvalue.

The following theorem holds for symmetric tensors:

\begin{theo}Let $\mathcal{A}\in\mathcal{S}_{m,n}$. If $m$ is even
then $\mathcal{A}$ always have H-eigenvalues and $\mathcal{A}$ is
positive definite (resp. semi-definite) if and only if its smallest
H-eigenvalue $\lambda_{H_{\min}}\left(\mathcal{A}\right)$ is positive
(resp. non-negative).

\end{theo}

\begin{theo}

Let $\mathcal{A}\in T_{m,n}$ be a nonnegative tensor. Then $\mathcal{A}$
has at least one H-eigenvalue and $\lambda_{H_{\max}}\left(\mathcal{A}\right)=\rho\left(\mathcal{A}\right)$.
Furthermore $\lambda_{H_{\max}}\left(\mathcal{A}\right)$ has a non-negative
H-eigenvector.

\end{theo}

\begin{defin}

A tensor is called an \textbf{essentially nonnegative tensor} if all
its off-diagonal entries are nonnegative.

A tensor is called a Z-tensor if its off-diagonal entries are nonpositive.

\end{defin}

\begin{theo}

Essentially nonnegative tensor and Z-tensors always have H-eigenvalues.

\end{theo}

\begin{defin}

Let $\mathcal{A}=\left(a_{i_{1}...i_{m}}\right)\in T_{m,n}$.

The diagonal elements of $\mathcal{A}$ are the elements $a_{i...i}$
for $i\in\left\llbracket 1,n\right\rrbracket $.

The off-diagonal elements of $\mathcal{A}$ are the other elements.

\end{defin}

An important result is the following:

\begin{prop}

Let $\mathcal{A}\in T_{m,n}$. Then the eigenvalues of $\mathcal{A}$
belongs to the union of $n$disks in $\mathbb{C}$. These $n$disks
have the diagonal entries of $\mathcal{A}$ as their centers and the
sums of the absolute values of the off-diagonal entries as their radii.

\end{prop}

\begin{rmk}

The proof shows that if $\left(\lambda,x\right)$ is an eigenpair
of $\mathcal{A}=\left(a_{i_{1}...i_{m}}\right)$, it holds for $i$
such that: $\left|x_{i}\right|=\max\left\{ \left|x_{j}\right|:1\leqslant j\leqslant n\right\} $:

\begin{equation}
\left|\lambda-a_{i...i}\right|\leqslant\sum\limits _{\substack{i_{2},...,i_{m}=1\\
\delta_{ii_{2}...i_{m}=0}
}
}^{n}\left|a_{ii_{2}...i_{m}}\right|\label{eq:bound_eigenval}
\end{equation}

\end{rmk}

\begin{corol}

If $\mathcal{A}$ is a nonnegative tensor of $T_{m,n}$ with an equal
row sum $r$. Then $r$ is the spectral radius of $\mathcal{A}$.

\end{corol}

\subsubsection{Spectral analysis of e-adjacency tensor}

Let $\mathcal{H}=\left(V,E\right)$ be a general hypergraph of e-adjacency
tensor $\mathcal{A}_{\mathcal{H}}=\left(a_{i_{1}...i_{k_{\max}}}\right)$

In the e-adjacency tensor $\mathcal{A}_{\mathcal{H}}$ built, the
diagonal entries are equal to zero. As all elements of $\mathcal{A}_{\mathcal{H}}$
are all non-negative real numbers and as we have shown that: 
\[
\sum\limits _{\substack{i_{2},...,i_{k_{\max}}=1\\
\delta_{ii_{2}...i_{k_{\max}}=0}
}
}^{n+k_{\max}-1}a_{ii_{2}...i_{k_{\max}}}=\begin{cases}
d_{i} & \text{if}\,1\leqslant i\leqslant n\\
d_{n+i} & \text{if}\,1\leqslant i\leqslant k_{\max}-1.
\end{cases}
\]

It follows:

\begin{theo}The e-adjacency tensor $\mathcal{A}_{\mathcal{H}}=\left(a_{i_{1}...i_{k_{\max}}}\right)$
of a general hypergraph $\mathcal{H}=\left(V,E\right)$ has its eigenvalues
$\lambda$ such that: 
\begin{equation}
\left|\lambda\right|\leqslant\max\left(\Delta,\Delta^{\star}\right)\label{eq:bound_max_degree_layer}
\end{equation}
 where $\Delta=\underset{1\leqslant i\leqslant n}{\max}\left(d_{i}\right)$
and $\Delta^{\star}=\underset{1\leqslant i\leqslant k_{\max}-1}{\max}\left(d_{n+i}\right)$

\end{theo}

\begin{proof}

From \ref{eq:bound_eigenval} we can write as $a_{i...i}=0$ and $a_{ii_{2}...i_{k_{\max}}}$are
non-negative numbers, that for all $\lambda$ it holds: $\left|\lambda\right|\leqslant\sum\limits _{\substack{i_{2},...,i_{k_{\max}}=1\\
\delta_{ii_{2}...i_{k_{\max}}=0}
}
}^{n+k_{\max}-1}a_{ii_{2}...i_{k_{\max}}}$ and thus writing $\Delta=\underset{1\leqslant i\leqslant n}{\max}\left(d_{i}\right)$
and $\Delta^{\star}=\underset{1\leqslant i\leqslant k_{\max}-1}{\max}\left(d_{n+i}\right)$,
it holds: $\left|\lambda\right|\leqslant\max\left(\Delta,\Delta^{\star}\right)$

\end{proof}

\begin{prop}Let $\mathcal{H}$ be a $r$-regular $r$-uniform hypergraph.
Then this maximum is reached.

\end{prop}

\begin{proof}In this case: 
\[
\forall1\leqslant i\leqslant n,\,d_{i}=\Delta=r
\]
 and 
\[
\Delta^{\star}=0,
\]
also: 
\[
\max\left(\Delta,\Delta^{\star}\right)=r.
\]

Considering $\lambda=r$ and the vector $\boldsymbol{1}$ which components
are only 1, $(r,\boldsymbol{1})$ is an eigenpair of $\mathcal{A_{H}}$
as forall $1\leqslant i\leqslant n$: 

\[
\begin{array}{crcl}
 & \sum\limits _{i_{2},...,i_{k_{\max}}=1}^{n+k_{\max}-1}a_{ii_{2}...i_{k_{\max}}}x_{i_{2}}...x_{i_{k_{\max}}} & = & \lambda x_{i}^{k_{\max}-1}\\
\Leftrightarrow & \sum\limits _{i_{2},...,i_{k_{\max}}=1}^{n+k_{\max}-1}a_{ii_{2}...i_{k_{\max}}} & = & r\\
\Leftrightarrow & d_{i} & = & r
\end{array}
\]

\end{proof}

\begin{rmk}

We see that this bound includes $\Delta^{\star}$ which can be close
to the number of hyperedges, for instance where the hyperedges would
be constituted of only one vertex per hyperedge except one hyperedge
with $k_{\max}\neq1$ vertices in it.

\end{rmk}

\section{Evaluation}

We have gathered some key features of both the e-adjacency tensor
proposed by \citet{banerjee2017spectra} - written $\mathcal{B}_{\mathcal{H}}$
and the one constructed in this article - written $\mathcal{A}_{\mathcal{H}}$.
The constructed tensor has same order. The dimension of $\mathcal{A_{H}}$
is $k_{\max}-1$ bigger than $\mathcal{B_{\mathcal{H}}}$ ($n-1$
in the worst case). The way $\mathcal{A}_{\mathcal{H}}$ is built
uses potentially $\dfrac{\left(n+k_{\max}-1\right)!}{(n-1)!n^{k_{\max}}}$
times less elements than for $\mathcal{B}_{\mathcal{H}}$ ($O\left(\dfrac{n!}{n^{n}}\right)$
in the worst case). The number of non-nul elements filled in $\mathcal{A_{\mathcal{H}}}$
for a given hypergraph is $\left(1+\dfrac{k_{\max}-1}{n}\right)^{k_{\max}}$
times the number of elements of $\mathcal{B_{\mathcal{H}}}$ ($O\left(4^{n}\right)$
times in the worst case). But the number of elements to be filled
to have full description of a hyperedge of size $s$ by permutation
of indices due to the symmetry of the tensor is only 1 in the case
of $\mathcal{A}_{H}$, which is $\dfrac{1}{p_{s}\left(k_{\max}\right)}$
times less than for a hyperedge stored in $\mathcal{B_{\mathcal{H}}}$.
The minimum number of elements needed to be described the other being
obtained by permutation is $\dfrac{1}{\left|E\right|}\sum\limits _{s=1}^{k_{\max}}p_{s}\left(k_{\max}\right)\left|E_{s}\right|$
bigger for $\mathcal{B}_{\mathcal{H}}$ than for $\mathcal{A}_{\mathcal{H}}$.
Moreover the value of the elements in $\mathcal{B_{\mathcal{H}}}$
varies with the cardinality of the hyperedge; in $\mathcal{A_{H}}$,
any element has same value. Both tensors allow the reconstruction
of the original hypergraph; for $\mathcal{B_{H}}$ it requires at
least $p_{s}\left(k_{\max}\right)$ check per hyperedge as for $\mathcal{A_{H}}$it
requires only one element per hyperedge.

In both cases, nodes degree can be deduced from the e-adjacency tensor.
$\mathcal{A}_{\mathcal{H}}$ allows the retrieval the structure of
the hypergraph in term of edges cardinality which is not straightforward
in the case of $\mathcal{B_{\mathcal{H}}}$.

The interpretability of $\mathcal{A}_{\mathcal{H}}$ in term of hypergraphs
is possible as it is the e-adjacency tensor of the $V_{s}$-layered
uniform hypergraph $\widehat{\mathcal{H}_{\widehat{w}}}$ obtained
from $\mathcal{H}$. $\mathcal{B}_{\mathcal{H}}$ is not interpretable
in term of hypergraphs, as hyperedges don't allow repetition of vertices.

\begin{table}
\begin{center}%
\begin{tabular}{|>{\centering}p{5cm}|>{\centering}m{3.7cm}|>{\centering}m{3cm}|}
\cline{2-3} 
\multicolumn{1}{>{\centering}p{5cm}|}{} & $\mathcal{B}_{\mathcal{H}}$ & $\mathcal{\mathcal{A_{\mathcal{H}}}}$\tabularnewline
\hline 
Order & $k_{\max}$ & $k_{\max}$\tabularnewline
\hline 
Dimension & $n$ & $n+k_{\max}-1$\tabularnewline
\hline 
Total number of elements & $n^{k_{\max}}$ & $\left(n+k_{\max}-1\right)^{k_{\max}}$\tabularnewline
\hline 
Total number of elements potentially used by the way the tensor is
build & $n^{k_{\max}}$ & $\dfrac{\left(n+k_{\max}-1\right)!}{\left(n-1\right)!}$\tabularnewline
\hline 
Number of non-nul elements for a given hypergraph & $\sum\limits _{s=1}^{k_{\max}}\alpha_{s}\left|E_{s}\right|$ with\newline$\alpha_{s}=p_{s}\left(k_{\max}\right)\dfrac{k_{\max}!}{k_{1}!...k_{s}!}$  & $k_{\max}!\left|E\right|$\tabularnewline
\hline 
Number of repeated elements per hyperedge of size $s$ & $\dfrac{k_{\max}!}{k_{1}!...k_{s}!}$ & $k_{\max}!$\tabularnewline
\hline 
Number of elements to be filled per hyperedge of size $s$ before
permutation & Varying

$p_{s}\left(k_{\max}\right)$ (\ref{Tab: Number of partitions of size s of an integer m}) & Constant

1\tabularnewline
\hline 
Number of elements to be described to derived the tensor by permutation
of indices & $\sum\limits _{s=1}^{k_{\max}}p_{s}\left(k_{\max}\right)\left|E_{s}\right|$ & $\left|E\right|$\tabularnewline
\hline 
Value of elements of a hyperedge & Varying\\
$\dfrac{s}{\alpha_{s}}$ & Constant\\
$\dfrac{1}{\left(k_{\max}-1\right)!}$\tabularnewline
\hline 
Symmetric & Yes & Yes\tabularnewline
\hline 
Reconstructivity & Need computation of duplicated vertices & Straightforward: delete special vertices\tabularnewline
\hline 
Nodes degree & Yes & Yes\tabularnewline
\hline 
Spectral analysis & Yes & Special vertices increase the amplitude of the bounds \tabularnewline
\hline 
Interpretability of the tensor in term of hypergraph & No & Yes\tabularnewline
\hline 
\end{tabular}

\caption{Evaluation of the e-adjacency tensor}

$\mathcal{B}_{\mathcal{H}}$ designates the adjacency tensor defined
in \citet{banerjee2017spectra}

$\mathcal{A_{H}}$ designates the layered e-adjacency tensor as defined
in this article.

\end{center}
\end{table}

\section{Future work and Conclusion}

\label{sec:Future-work-and}

The importance of defining properly the concept of adjacency in a
hypergraph has helped us to build a proper e-adjacency tensor in a
way that allows to contain important information on the structure
of the hypergraph. This work contributes to give a methodology to
build a uniform hypergraph and hence a cubical symmetric tensor from
the different layers of uniform hypergraphs contained in a hypergraph.
The built tensor allows to reconstruct with no ambiguity the original
hypergraph. Nonetheless, first results on spectral analysis show difficulties
to use the tensor built as the additional vertices inflate the spectral
radius bound. The uniformisation process is a strong basis for building
alternative proposals.

\section{Acknowledgments}

This work is part of the PhD of Xavier OUVRARD, done at UniGe, supervised
by Stéphane MARCHAND-MAILLET and founded by a doctoral position at
CERN, in Collaboration Spotting team, supervised by Jean-Marie LE
GOFF.

The authors are really thankful to all the team of Collaboration Spotting:
Adam AGOCS, Dimitris DARDANIS, Richard FORSTER and a special thanks
to Andre RATTINGER for the daily long exchanges we have on our respective
PhD.

\bibliographystyle{plainnat}
\bibliography{/home/xo/cernbox/these/000-thesis_corpus/biblio/references}

\begin{thebibliography}{19}
\providecommand{\natexlab}[1]{#1}
\providecommand{\url}[1]{\texttt{#1}}
\expandafter\ifx\csname urlstyle\endcsname\relax
  \providecommand{\doi}[1]{doi: #1}\else
  \providecommand{\doi}{doi: \begingroup \urlstyle{rm}\Url}\fi

\bibitem[Banerjee et~al.(2017)Banerjee, Char, and Mondal]{banerjee2017spectra}
Anirban Banerjee, Arnab Char, and Bibhash Mondal.
\newblock Spectra of general hypergraphs.
\newblock \emph{Linear Algebra and its Applications}, 518:\penalty0 14--30,
  2017.

\bibitem[Berge and Minieka(1973)]{berge1973graphs}
Claude Berge and Edward Minieka.
\newblock \emph{Graphs and hypergraphs}, volume~7.
\newblock North-Holland publishing company Amsterdam, 1973.

\bibitem[Bretto(2013)]{bretto2013hypergraph}
Alain Bretto.
\newblock Hypergraph theory.
\newblock \emph{An introduction. Mathematical Engineering. Cham: Springer},
  2013.
\newblock \doi{10.1007/978-3-319-00080-0}.
\newblock URL \url{http://dx.doi.org/10.1007/978-3-319-00080-0}.

\bibitem[Bruckstein et~al.(2009)Bruckstein, Donoho, and
  Elad]{bruckstein2009sparse}
Alfred~M Bruckstein, David~L Donoho, and Michael Elad.
\newblock From sparse solutions of systems of equations to sparse modeling of
  signals and images.
\newblock \emph{SIAM review}, 51\penalty0 (1):\penalty0 34--81, 2009.

\bibitem[Comon et~al.(2008)Comon, Golub, Lim, and Mourrain]{comon2008symmetric}
Pierre Comon, Gene Golub, Lek-Heng Lim, and Bernard Mourrain.
\newblock Symmetric tensors and symmetric tensor rank.
\newblock \emph{SIAM Journal on Matrix Analysis and Applications}, 30\penalty0
  (3):\penalty0 1254--1279, 2008.

\bibitem[Comon et~al.(2015)Comon, Qi, and Usevich]{comon2015polynomial}
Pierre Comon, Yang Qi, and Konstantin Usevich.
\newblock A polynomial formulation for joint decomposition of symmetric tensors
  of different orders.
\newblock In \emph{International Conference on Latent Variable Analysis and
  Signal Separation}, pages 22--30. Springer, 2015.

\bibitem[Cooper and Dutle(2012)]{cooper2012spectra}
Joshua Cooper and Aaron Dutle.
\newblock Spectra of uniform hypergraphs.
\newblock \emph{Linear Algebra and its Applications}, 436\penalty0
  (9):\penalty0 3268--3292, 2012.

\bibitem[Dewar et~al.(2016)Dewar, Pike, and Proos]{dewar2016connectivity}
Megan Dewar, David Pike, and John Proos.
\newblock Connectivity in hypergraphs.
\newblock \emph{arXiv preprint arXiv:1611.07087}, 2016.

\bibitem[Ghoshdastidar and Dukkipati(2017)]{ghoshdastidar2017uniform}
Debarghya Ghoshdastidar and Ambedkar Dukkipati.
\newblock Uniform hypergraph partitioning: Provable tensor methods and sampling
  techniques.
\newblock \emph{Journal of Machine Learning Research}, 18\penalty0
  (50):\penalty0 1--41, 2017.

\bibitem[Hu(2013)]{hu2013spectral}
Shenglong Hu.
\newblock \emph{Spectral hypergraph theory}.
\newblock PhD thesis, The Hong Kong Polytechnic University, 2013.

\bibitem[Lim(2013)]{lim2013tensors}
Lek-Heng Lim.
\newblock Tensors and hypermatrices.
\newblock \emph{Handbook of Linear Algebra, 2nd Ed., CRC Press, Boca Raton,
  FL}, pages 231--260, 2013.

\bibitem[Michoel and Nachtergaele(2012)]{michoel2012alignment}
Tom Michoel and Bruno Nachtergaele.
\newblock Alignment and integration of complex networks by hypergraph-based
  spectral clustering.
\newblock \emph{Physical Review E}, 86\penalty0 (5):\penalty0 056111, 2012.

\bibitem[Nikolova(2000)]{nikolova2000}
Mila Nikolova.
\newblock Local strong homogeneity of a regularized estimator.
\newblock \emph{SIAM Journal on Applied Mathematics}, 61\penalty0 (2):\penalty0
  633--658, 2000.
\newblock ISSN 00361399.
\newblock URL \url{http://www.jstor.org/stable/3061742}.

\bibitem[Ouvrard and Marchand-Maillet(2018)]{ouvrard2018hypergraph_survey}
Xavier Ouvrard and St{\'e}phane Marchand-Maillet.
\newblock Hypergraphs: a survey.
\newblock \emph{Soon on Arxiv}, 2018.

\bibitem[Pearson and Zhang(2014)]{Pearson2014}
Kelly~J. Pearson and Tan Zhang.
\newblock On spectral hypergraph theory of the adjacency tensor.
\newblock \emph{Graphs and Combinatorics}, 30\penalty0 (5):\penalty0
  1233--1248, Sep 2014.
\newblock ISSN 1435-5914.
\newblock \doi{10.1007/s00373-013-1340-x}.
\newblock URL \url{https://doi.org/10.1007/s00373-013-1340-x}.

\bibitem[Pu(2013)]{pu2013relational}
Li~Pu.
\newblock Relational learning with hypergraphs.
\newblock 2013.

\bibitem[Qi and Luo(2017)]{qi2017siam}
L.~Qi and Z.~Luo.
\newblock \emph{Tensor Analysis}.
\newblock Society for Industrial and Applied Mathematics, 2017.
\newblock \doi{10.1137/1.9781611974751}.
\newblock URL \url{http://epubs.siam.org/doi/abs/10.1137/1.9781611974751}.

\bibitem[Shao(2013)]{shao2013general}
Jia-Yu Shao.
\newblock A general product of tensors with applications.
\newblock \emph{Linear Algebra and its applications}, 439\penalty0
  (8):\penalty0 2350--2366, 2013.

\bibitem[Zhou et~al.(2007)Zhou, Huang, and Sch{\"o}lkopf]{zhou2007learning}
Denny Zhou, Jiayuan Huang, and Bernhard Sch{\"o}lkopf.
\newblock Learning with hypergraphs: Clustering, classification, and embedding.
\newblock In \emph{Advances in neural information processing systems}, pages
  1601--1608, 2007.

\end{thebibliography}

\section*{Appendix A\label{sec:Appendix-A}}

Example

Given the following hypergraph: $\mathcal{H}=\left(V,E\right)$ where:
$V=\left\{ v_{1},v_{2},v_{3},v_{4},v_{5},v_{6},v_{7}\right\} $ and
$E=\left\{ e_{1},e_{2},e_{3},e_{4},e_{5},e_{6},e_{7}\right\} $ with:
$e_{1}=\left\{ v_{1},v_{2},v_{3}\right\} $, $e_{2}=\left\{ v_{1},v_{2},v_{7}\right\} $,
$e_{3}=\left\{ v_{6},v_{7}\right\} $, $e_{4}=\left\{ v_{5}\right\} $,
$e_{5}=\left\{ v_{4}\right\} $, $e_{6}=\left\{ v_{3},v_{4}\right\} $
and $e_{7}=\left\{ v_{4},v_{7}\right\} $.

This hypergraph $\mathcal{H}$ is drawn in Figure \ref{Fig: Layers of an hypergraph}.

The layers of $\mathcal{H}$ are:
\begin{itemize}
\item $\mathcal{H}_{1}=\left(V,\left\{ e_{4},e_{5}\right\} \right)$ with
the associated unnormalized tensor: 
\[
\mathcal{A}_{1\,raw}=\left[\begin{array}{ccccccc}
0 & 0 & 0 & 1 & 1 & 0 & 0\end{array}\right]
\]
 and associated homogeneous polynomial: 
\[
P_{1}\left(z_{0}\right)=z_{4}+z_{5}.
\]
More generally, the version with a normalized tensor is: 
\[
P_{1}\left(z_{0}\right)=a_{(1)\,4}z_{4}+a_{(1)\,5}z_{5}
\]
\item $\mathcal{H}_{2}=\left(V,\left\{ e_{3},e_{6},e_{7}\right\} \right)$
with the associated unnormalized tensor: 
\[
\mathcal{A}_{2\,raw}=\left[\begin{array}{ccccccc}
0 & 0 & 0 & 0 & 0 & 0 & 0\\
0 & 0 & 0 & 0 & 0 & 0 & 0\\
0 & 0 & 0 & 1 & 0 & 0 & 0\\
0 & 0 & 1 & 0 & 0 & 0 & 1\\
0 & 0 & 0 & 0 & 0 & 0 & 0\\
0 & 0 & 0 & 0 & 0 & 0 & 1\\
0 & 0 & 0 & 1 & 0 & 1 & 0
\end{array}\right]
\]
 and associated homogeneous polynomial: 
\[
P_{2}\left(z_{0}\right)=2z_{3}z_{4}+2z_{6}z_{7}+2z_{4}z_{7}.
\]
More generally, the version with a normalized tensor is: 
\[
P_{2}\left(z_{0}\right)=2!a_{(2)\,3\,4}z_{3}z_{4}+2!a_{(2)\,6\,7}z_{6}z_{7}+2!a_{(2)\,4\,7}z_{4}z_{7}.
\]
\item $\mathcal{H}_{3}=\left(V,\left\{ e_{1},e_{2}\right\} \right)$ with
the associated unnormalized tensor: \\
\resizebox{.9\hsize}{!}{$\mathcal{A}_{3\,raw}=\left[\left.\begin{array}{ccccccc} 0 & 0 & 0 & 0 & 0 & 0 & 0\\ 0 & 0 & 1 & 0 & 0 & 0 & 1\\ 0 & 1 & 0 & 0 & 0 & 0 & 0\\ 0 & 0 & 0 & 0 & 0 & 0 & 0\\ 0 & 0 & 0 & 0 & 0 & 0 & 0\\ 0 & 0 & 0 & 0 & 0 & 0 & 0\\ 0 & 1 & 0 & 0 & 0 & 0 & 0 \end{array}\right|\left.\begin{array}{ccccccc} 0 & 0 & 1 & 0 & 0 & 0 & 1\\ 0 & 0 & 0 & 0 & 0 & 0 & 0\\ 1 & 0 & 0 & 0 & 0 & 0 & 0\\ 0 & 0 & 0 & 0 & 0 & 0 & 0\\ 0 & 0 & 0 & 0 & 0 & 0 & 0\\ 0 & 0 & 0 & 0 & 0 & 0 & 0\\ 1 & 0 & 0 & 0 & 0 & 0 & 0 \end{array}\right|\left.\begin{array}{ccccccc} 0 & 1 & 0 & 0 & 0 & 0 & 0\\ 1 & 0 & 0 & 0 & 0 & 0 & 0\\ 0 & 0 & 0 & 0 & 0 & 0 & 0\\ 0 & 0 & 0 & 0 & 0 & 0 & 0\\ 0 & 0 & 0 & 0 & 0 & 0 & 0\\ 0 & 0 & 0 & 0 & 0 & 0 & 0\\ 0 & 0 & 0 & 0 & 0 & 0 & 0 \end{array}\right|\left.0\right|\left.0\right|\left.0\right|\left.\begin{array}{ccccccc} 0 & 1 & 0 & 0 & 0 & 0 & 0\\ 1 & 0 & 0 & 0 & 0 & 0 & 0\\ 0 & 0 & 0 & 0 & 0 & 0 & 0\\ 0 & 0 & 0 & 0 & 0 & 0 & 0\\ 0 & 0 & 0 & 0 & 0 & 0 & 0\\ 0 & 0 & 0 & 0 & 0 & 0 & 0\\ 0 & 0 & 0 & 0 & 0 & 0 & 0 \end{array}\right|\right]$}\\
 and associated homogeneous polynomial: 
\[
P_{3}\left(z_{0}\right)=3!z_{1}z_{2}z_{3}+3!z_{1}z_{2}z_{7}
\]
More generally, the version with a normalized tensor is:
\[
P_{3}\left(z_{0}\right)=3!a_{(3)\,1\,2\,3}z_{1}z_{2}z_{3}+3!a_{(3)\,1\,2\,7}z_{1}z_{2}z_{7}.
\]
\end{itemize}
The iterative process of homogenization is then the following using
the degree-normalized adjacency tensor $\mathcal{A}_{k}=\dfrac{1}{(k-1)!}\mathcal{A}_{k\,raw}$
and the normalizing coefficients $c_{k}=\dfrac{k_{\max}}{k}$, with
$k_{\max}=3$
\begin{itemize}
\item $R_{1}\left(\boldsymbol{z_{0}}\right)=\dfrac{3}{1}P_{1}\left(z_{0}\right)$
\item $R_{2}\left(\boldsymbol{z_{1}}\right)=R_{1}\left(z_{0}\right)y_{1}+\dfrac{3}{2}P_{2}\left(z_{0}\right)$
\item $R_{3}\left(\boldsymbol{z_{2}}\right)=R_{2}\left(z_{1}\right)y_{2}+\dfrac{3}{3}P_{3}\left(z_{0}\right)$
\end{itemize}
Hence:

\begin{eqnarray*}
R_{3}\left(\boldsymbol{z_{2}}\right) & = & 3\left(a_{(1)\,4}z_{4}+a_{(1)\,5}z_{5}\right)y_{1}y_{2}\\
 &  & +\dfrac{3}{2}\times2!\left(a_{(2)\,3\,4}z_{3}z_{4}+a_{(2)\,6\,7}z_{6}z_{7}+a_{(2)\,4\,7}z_{4}z_{7}\right)y_{2}\\
 &  & +3!\left(a_{(3)\,1\,2\,3}z_{1}z_{2}z_{3}+a_{(3)\,1\,2\,7}z_{1}z_{2}z_{7}\right).
\end{eqnarray*}

Therefore the e-adjacency tensor of $\mathcal{H}$ is obtained by
writing the corresponding symmetric cubical tensor of order 3 and
dimension 9, described by: $r_{489}=r_{589}=r_{349}=r_{679}=r_{479}=r_{123}=r_{127}=\dfrac{1}{2}$.
The other remaining not null elements are obtained by permutation
on the indices.

Finding the degree of one vertex from the tensor is easily achievable;
for instance $\deg\left(v_{4}\right)=2!\left(r_{489}+r_{349}+r_{479}\right)=3$.
\end{document}